\documentclass[12pt]{amsart}
\usepackage{amsmath,amssymb}
\usepackage{amsmath}
\usepackage{amsopn}
\usepackage{fixltx2e}
\usepackage{url}
\usepackage{hyperref}

\numberwithin{equation}{section} %% Comment out for sequentially-numbered
\numberwithin{figure}{section} %% Comment out for sequentially-numbered

\theoremstyle{plain}
\newtheorem{thm}{Theorem}
\newtheorem{cor}[thm]{Corollary} %%Delete [thm] to re-start numbering
\newtheorem{prob}{Problem}
\newtheorem{lem}{Lemma}[section] %%Delete [thm] to re-start numbering
\newtheorem{prop}[thm]{Proposition} %%Delete [thm] to re-start numbering

\newtheorem*{lemW}{Lemma W}
\newtheorem*{lemF}{Lemma F}

\theoremstyle{definition}
\newtheorem{defn}{Definition}
\newtheorem{ass}{Assumption}

\theoremstyle{remark}
 \newtheorem*{rem*}{Remark}
 \newtheorem*{rem}{Remark}
 \newtheorem*{rems}{Remarks}

\newcommand{\ip}[1]{\left\langle #1 \right \rangle}
\newcommand{\eq}[1]{eq.~(\ref{#1})}  %% produces: eq.(#)
\newcommand{\Ev}[1]{\E \left( #1 \right)}  %% produces \E( # )
\newcommand{\norm}[1]{\left\Vert#1\right\Vert}
\newcommand{\abs}[1]{\left\vert#1\right\vert}
\newcommand{\set}[1]{\left\{#1\right\}}

\renewcommand{\vec}[1]{\mathbf{#1}}
\newcommand{\field}[1]{\mathbb{#1}}
\newcommand{\bb}[1]{\mathbb{#1}}
\newcommand{\cu}[1]{\mathcal{#1}}

\def\Var{\operatorname{Var}}

\def\const{\mathrm{const.}}

\def\e{\mathrm e}

\def\im{\mathrm i}
\def\Im{\mathrm{Im}}
\def\half {\frac{1}{2}}
\def\1{{\mathsf 1}}
\def\di{\mathrm d}

\makeatletter
\def\rightharpoondownfill@{%
    \arrowfill@\relbar\relbar\rightharpoondown}
\def\rightharpoonupfill@{%
    \arrowfill@\relbar\relbar\rightharpoonup}
\def\leftharpoondownfill@{%
    \arrowfill@\leftharpoondown\relbar\relbar}
\def\leftharpoonupfill@{%
    \arrowfill@\leftharpoonup\relbar\relbar}
\newcommand{\xrightharpoondown}[2][]{%
    \ext@arrow 0359\rightharpoondownfill@{#1}{#2}}
\newcommand{\xrightharpoonup}[2][]{%
    \ext@arrow 0359\rightharpoonupfill@{#1}{#2}}
\newcommand{\xleftharpoondown}[2][]{%
    \ext@arrow 3095\leftharpoondownfill@{#1}{#2}}
\newcommand{\xleftharpoonup}[2][]{%
    \ext@arrow 3095\leftharpoonupfill@{#1}{#2}}
\newcommand{\xleftrightharpoons}[2][]{\mathrel{%
    \raise.22ex\hbox{%
        $\ext@arrow 3095\leftharpoonupfill@{\phantom{#1}}{#2}$}%
    \setbox0=\hbox{%
        $\ext@arrow 0359\rightharpoondownfill@{#1}{\phantom{#2}}$}%
    \kern-\wd0 \lower.22ex\box0}%
}
\newcommand{\xrightleftharpoons}[2][]{\mathrel{%
    \raise.22ex\hbox{%
        $\ext@arrow 3095\rightharpoonupfill@{\phantom{#1}}{#2}$}%
    \setbox0=\hbox{%
        $\ext@arrow 0359\leftharpoondownfill@{#1}{\phantom{#2}}$}%
    \kern-\wd0 \lower.22ex\box0}%
} \makeatother

\def\N{\mathbb N}
\def\R{\mathbb R}
\def\C{\mathbb C}
\def\E{\mathbb E}

\def\ra{\rightarrow}

\def\Pr{\operatorname{Prob}} %%% appears in many equations  Prob

\def\dist{\operatorname{dist}}   %%%  distance
\def\dim{\operatorname{dim}}  %%% dimension
\def\tr{\operatorname{tr}}    %%% Trace
\def\esssup{\operatorname*{ess-sup}}

\def\Re{\operatorname{Re}}
\def\Im{\operatorname{Im}}
\def\tem{\textemdash \ }

\def\lf{\eta}
\title[Localization for random band matrices]{Eigenvector localization for random band matrices with power law band width}
\author[J. Schenker]{Jeffrey Schenker}
\address{Michigan State University \\ East Lansing, Michigan 48824}
\email{jeffrey@math.msu.edu}
\date{September 25, 2008; revised January 07, 2009}

\begin{document}
\maketitle
\begin{abstract} It is shown that certain ensembles of random matrices with entries that vanish outside a band around the diagonal satisfy a localization condition on the resolvent which guarantees that eigenvectors have strong overlap with a vanishing fraction of standard basis vectors, provided the band width $W$ raised to a power $\mu$ remains smaller than the matrix size $N$.  For a Gaussian band ensemble, with matrix elements given by i.i.d.\ centered Gaussians within a band of width $W$, the estimate $\mu \le 8$ holds.
\end{abstract}

\bibliographystyle{amsplain}

\section{Introduction}
Random band matrices, with entries that vanish outside a band of width
$W$ around the diagonal, have been suggested \cite{Casati:1990p3447, Casati:1993p35} as a model to study the crossover between a strongly disordered ``insulating'' regime,  with localized eigenfunctions and weak eigenvalue correlations, and a weakly disordered ``metallic'' regime, with extend eigenfunctions and strong eigenvalue repulsion. Such a crossover is believed to occur in the spectra of certain random partial differential (or difference) operators as the spectral parameter (energy) is changed.  

In this paper, the strong disorder side of the band matrix crossover is analyzed. It is shown here that certain ensembles of random matrices whose entries vanish in a band of width $W$ around the diagonal satisfy a localization condition  in the limit that the size of the matrix $N$  tends to infinity provided $W^{8}/N \ra 0$.   This result requires
certain assumptions on the distribution of the entries of the
matrix, and the proof given here has technical requirements that
may not be necessary. Nonetheless, the conditions imposed below (see
\S\ref{sec:axioms}) allow for a large family
of interesting examples.  In particular, one may consider a Gaussian
distributed band matrix, with distribution
\begin{equation}\label{eq:GBE}
     \mathrm{e}^{-2W \tr X_{W;N}^2 } \di X_{W;N}
\end{equation}
where $\di X_{W;N}$ the Lebesgue measure on the vector space of
$N\times N$ matrices of band width $W$. That is
\begin{equation}\label{eq:bandmatrix}
X_{W;N} \ = \ \frac{1}{\sqrt{W}}  \underbrace{\begin{pmatrix}
       d_{1,1} & a_{1,2}  \cdots & a_{1,W}  \\
       a_{2,1}^*& d_{2,2} & &\ddots  \\
       \vdots & & \ddots & & \ddots\\
       a_{W,1} & & &  \ddots & & \ddots \\
       & \ddots & & & \ddots & & \ddots \\
       & & \ddots & & & \ddots && \ddots
     \end{pmatrix}}_{N \times N},
\end{equation}
with $d_i$ and $a_{i,j}$ independent families of i.i.d.\ real and
complex unit Gaussian variables, respectively.  

The main result obtained here is a localization estimate for the eigenvectors of the matrices $X_{W;N}$.  This localization result is most conveniently stated in terms of the resolvent $(X_{W;N} -\lambda)^{-1}$, a well defined random matrix for $\lambda \in\R$. (We will see that $\lambda$ is an eigenvalue of $X_{W;N}$ with probability zero.) Let $\mathbf{e}_i$ denote the standard basis vectors
$\mathbf{e}_i(j) = \delta_{i,j}$. Then 
\begin{thm}\label{thm:resolvent}  If $X_{W;N}$ has distribution  \eqref{eq:GBE}, or more generally a
distribution satisfying assumptions 1, 2, and 3 in \S\ref{sec:axioms} below, then there exists $
\mu >0$ and $\sigma < \infty$ such that given $r >0$ and $s  \in (0,1)$ there are $A_s < \infty$ and $\alpha_s > 0$ 
such
that
\begin{equation}\label{eq:resolvloc}
  \Ev{\abs{\ip{\mathbf{e}_i, (X_{W;N} - \lambda)^{-1}
  \mathbf{e}_j}}^s} \ \le \ A_s W^{s\sigma} \e^{-\alpha_s \frac{|i-j|}{W^{ \mu}}} 
\end{equation}
for all $\lambda \in [-r,r]$ and all $i,j=1,\ldots, N$. For the Gaussian band ensemble \eqref{eq:GBE} $\sigma \le \half$ and  $\mu \le 8$.
\end{thm}
\noindent \emph{Remarks}: For the Gaussian Band Ensemble \eqref{eq:GBE}, the density of states, in the regime $W, N \ra \infty$, $W/N \ra 0$, is known to be the Wigner semi-circle law (see \S\ref{sec:ensembles} below).  For $\lambda$ outside the support of the semi-circle law, one could obtain \eqref{eq:resolvloc} with $\mu=1$ using Lifschitz tail type estimates.  This will be dealt with in a separate paper.

Theorem \ref{thm:resolvent} estimates the decay of matrix elements of the resolvent away from the diagonal.  Using techniques developed in the context of discrete random Schr\"odinger operators one may obtain from \eqref{eq:resolvloc} estimates on eigenvectors.
\begin{thm}[Eigenvector localization]\label{thm:eigenvector}   Let $X_{W;N}$ have distribution  \eqref{eq:GBE}, or more generally a distribution satisfying assumptions 4 and 5 in \S\ref{sec:ensembles}.  \begin{enumerate}
\item With probability one all eigenvalues of $X_{W;N}$ are simple.
\item  If \eqref{eq:resolvloc}  holds for all $\lambda$ in an interval $[-r,r]$ and
if $\lambda_{k}$, $k=1,\ldots, N$, are the eigenvalues of $X_{W;N}$ with corresponding eigenvectors $\vec{v}_{k}$, $k=1, \ldots,N$, then there are $B < \infty$, $\tau \ge 0$, and $\beta >  0$ 
\begin{equation}\label{eq:eigloc}
\Ev{\sup_{ \lambda_{k} \in [-r,r] } \abs{\vec{v}_{k}(i) \vec{v}_{k}(j)}} \ \le \ B W^{\tau}\e^{-\beta \frac{|i-j|}{W^{\mu}}}
\end{equation}
for all $i,j=1, \ldots, N$.
 \end{enumerate}
 \end{thm}
 \begin{rem} For the proof of this theorem, the reader is directed to the corresponding result in the context of random Schroedinger operators.  See for example \cite{Klein:2006zl} for the non-degeneracy of the eigenvalues and \cite[Theorem A.1]{Aizenman:2001p2115}  for a derivation of \eqref{eq:eigloc} from Green's function decay \eqref{eq:resolvloc}.  In both cases, the proof involves only averaging over the coupling of a rank one perturbation and can be applied in the present context.
\end{rem}

\subsection{Sketch of the Proof}  The proof of Theorem \ref{thm:resolvent} is based on two
observations, which may be summarized as follows.\footnote{The idea to study localization via these two complementary estimates was suggested in the context of random Schroedinger operators by Michael Aizenman, and is inspired by the Dobrushin-Shlosman proof \cite{Dobrushin:1975qx} of the Mermin-Wagner Theorem \cite{Mermin:1966cs} on the absence of continuous symmetry breaking in classical statistical mechanics of dimension 2.}  Let $G_{W;N}(i,j)
= \ip{\mathbf{e}_i, (X_{W;N} - \lambda)^{-1}
  \mathbf{e}_j}$.  Then

\begin{enumerate}
\item \emph{The random variable $G_{W;N}(i,j)$ is rarely large.} This may be expressed
through a bound (uniform in $N$) on the tails of the distribution of $G_{W;N}(i,j)$
\begin{lemW}
\label{lem:wegner}If $X_{W;N}$ has distribution  \eqref{eq:GBE}, or more generally a
distribution with the properties outlined in \S\ref{sec:axioms} below, then there exist $\kappa>0$ and $\sigma < \infty$ such that
\begin{equation}\Pr \left ( |G_{W;N}(i,j)  | > t \right ) \le \kappa \
\frac{W^{\sigma}}{t} . \label{eq:1/ttails}\end{equation}
\end{lemW}

\item \emph{The fluctuations of $\ln|G_{W;N}(i,j)|$ grow at least linearly with $|i-j|$.} One would typically express the growth of fluctuations by an inequality like
$$ \Var( \ln|G_{W;N}(i,j)|) \ \ge \ \const \ |i-j| ,$$
where $\Var(X) = \E(X^{2}) - \E(X)^{2}$ is the variance of a random variable $X$.  However for present purposes a more convenient quantitative expression of this idea is the following
\begin{lemF}
\label{lem:fluct} If $X_{W;N}$ has distribution  \eqref{eq:GBE}, or more generally a
distribution with the properties outlined in \S\ref{sec:axioms} below, then there is $\nu > 0$ such that if $0 < r < s< 1$ and $|i-j| > 3 W$
then \begin{equation} \E\left(|G_{W;N}(i,j)|
^{r}\right)\le\exp(-C_{r,s}W^{-\mu} |i-j|)\E\left(|G_{W;N}(i,j)|
^{s}\right)^{r/s}\label{eq:fluctuations}\end{equation} with $C_{r,s}
>0$. For the Gaussian band ensemble \eqref{eq:GBE} $\mu \le 8$.
\end{lemF}
\end{enumerate}

Lemmas W and F together easily imply Theorem \ref{thm:resolvent}.  
Indeed, it suffices to show that the second factor on the right hand side of \eqref{eq:fluctuations} is uniformly bounded.  But it follows from Lemma W that
\begin{equation} \E\left(| G_{W;N}(i,j)|
^{s}\right)\leq\frac{\kappa^{s}}{1-s}W^{s \sigma}.\label{eq:smomentbound}\end{equation}
This observation, which is the basis of the fractional moment  analysis of random Schr\"odinger operators \cite{Aizenman:1994qm,Aizenman:1993p2946,Aizenman:2001p2115}, follows easily from \eqref{eq:1/ttails} since 
\begin{equation}\label{eq:bathtub}
\E\left(| G_{W;N}(i,j)|
^{s}\right) \ = \ s \int_{0}^{\infty} \Pr \left ( |G_{W;N}(i,j)  | > t \right ) t^{s-1} \di t,
\end{equation}
and  probabilities are bounded by one.

It may not be immediately clear what Lemma F has to do with large fluctuations.  Towards understanding this, let $X = \ln |G_{W;N}(i,j)|$.    By the H\"older inequality, \begin{equation}\label{eq:Holder} \E(\e^{r X}) \le \E(\e^{s X})^{\frac{r}{s}}.\end{equation} 
 Furthermore, 
equality holds only if $X$ is \emph{non random} \textemdash \ if there is $x_{0}\in \R$ so that $X=x_{0}$ almost surely. In other words
\begin{equation} \label{eq:improvedHolder}\E(\e^{rX}) = \e^{-h(r,s)} \E(\e^{s X})^{\frac{r}{s}}\end{equation}
with $h(r,s) > 0$ unless $X$ is non random.    

If $X$ were Gaussian with
variance $\sigma^{2}$ (and arbitrary mean), then $h(r,s)$ would be proportional to the variance 
\begin{equation}
h(r,s) \ =\ \frac{r(s-r)}{2} \sigma^2.\label{eq:gaussian}\end{equation}
For a general random variable $X$, the associated quantity $h(r,s)$ may be taken as a measure of the fluctuations of $X$.  In place of \eqref{eq:gaussian}, we 
 have the following identity for $h$ in terms of the variance of $X$ in weighted ensembles:
\begin{prop}\label{prop:basic}
  Let $X$ be a random variable with $\Ev{\e^{s X}} <
  \infty$ for some $s > 0$.  If $r \in (0,s)$, then $\Ev{\e^{r X}} < \infty$ and
  \begin{equation}\label{eq:cl3}
  h(r,s) \ = \ \frac{1}{s} \int_{0}^{s} \min(r,q) \left ( s - \max(r,q) \right ) 
  \Var_{q}(X) \di q,
  \end{equation}
  where $h(r,s)$ is defined by \eqref{eq:improvedHolder} and 
  \begin{equation}
  \Var_{q}(X) \ = \ \frac{\Ev{X^{2} \e^{q X}} }{\Ev{\e^{q X}}} - \left ( \frac{\Ev{X \e^{q X}}}{\Ev{\e^{q X}}} \right ) ^{2}
  \end{equation} is the variance of
  $X$ with respect to the weighted probability measure  $ \Pr_q(A) = \Ev{\chi_{A} \e^{q X}}/\Ev{\e^{q
X}}.$
\end{prop}
\begin{proof}
H\"older's inequality is the statement that the function $
  \Phi(r) =  \ln \Ev{\e^{r \sigma}}$
is convex. In particular, if $s > 0$ then
\begin{equation}\label{eq:cl1}
\Phi(r) \ \le \ \frac{r}{s} \Phi(s)
\end{equation}
for $r \in (0,s)$, since $\Phi(1) = \ln \Ev{1} = 0$. If $\E(\e^{s\sigma}) < \infty$, it follows that $\Phi$ is bounded on $[0,s]$.

The identity \eqref{eq:cl3} follows from Taylor's formula with remainder. Indeed, the second derivative of $\Phi$ at $r$ is equal to the weighted variance $\Var_r(X)$. Thus, \begin{equation}
\Phi(s) =  \Phi(r) + \Phi'(r) (s-r) +  \int_{r}^{s} (s-q) \Var_{q}(X) \di q, \quad \text{ and}
\end{equation}
\begin{equation}
0 = \Phi(0) = \Phi(r) - \Phi'(r)r + \int_{0}^{r} q \Var_{q}(X) \di q .
\end{equation}
Taking a convex combination of these identities, chosen so the first order terms cancel,  gives
\begin{multline}
\frac{r}{s} \Phi(s) \  = \  \Phi(r) + \int_{0}^{r} \frac{(s-r)q}{s}  \Var_{q}(X) \di q 
+ \int_r^s \frac{(s-q)r}{s} \Var_q(X) \di q \\ = \ \Phi(r) + \frac{1}{s} \int_0^s 
\min (r,q) ( s  - \max(r,q) ) 
\Var_q(X) \di q,
\end{multline}
which is equivalent to \eqref{eq:cl3}.
\end{proof}

Thus Lemma F may be understood as giving a lower bound on the fluctuations of $X = \ln |G_{W;N}(i,j)|$, as measured by the improvement to H\"older's inequality.  The proof of this result will be accomplished using a product formula for $G_{W;N}(i,j)$ that expresses this quantity as a matrix element of a product of $O(|i-j|/W)$ matrices of size $W \times W$.  Prop.\ \eqref{prop:basic} will be applied to factors in this product, with each factor contributing a term of size $1/W^{7}$ to $h(r,s)$.  Since there are $O(|i-j|/W)$ terms, this produces the claimed decay.

The strategy taken below in proving Lemmas W and F  has two parts.  First we identify certain axioms for the distribution of $X_{W;N}$ which lead naturally to the lemmas. Second, we verify that the Gaussian band ensemble \eqref{eq:GBE} satisfies these axioms. To motivate the form of the axioms for the distribution of $X_{W;N}$, we begin in \S\ref{sec:W=2} with a self contained sketch of the argument in the tri-diagonal case $W=2$.    In \S\ref{sec:axioms} we state the assumptions needed to adapt the proof to $W >2$, state the associated general results and prove Lemma W.  In \S\ref{sec:proof} we get to the heart of the matter and prove Lemma F.   In  \S\ref{sec:ensembles}, we discuss examples of ensembles, including the Gaussian band ensemble \eqref{eq:GBE}, satisfying the axioms of \S\ref{sec:axioms}.   In an appendix, an elementary probability lemma used below is stated and proved.

\subsection{Remarks on the literature and open problems}
In \cite{Casati:1990p3447, Casati:1993p35} it
was observed, based on numerical evidence, that the localization of eigenfunctions and eigenvalue statistics of the Gaussian band ensemble \eqref{eq:GBE} are essentially determined by the parameter $ W^2/N$.  When $W^2/N <<1$ the eigenfunctions are strongly localized and the eigenvalue process is close to a Poisson process.  When $W^2/N >> 1$ the eigenfunctions are extended and the eigenvalue statistics are well described by the Gaussian unitary ensemble (GUE).   A theoretical physics explanation of these numerical results was given by Fyodorov and Mirlin \cite{Fyodorov:1991p5509}.  They considered a slightly different ensemble in which a full GUE matrix is modified by multiplying each
element by a factor which decays exponentially in the distance from the diagonal. For this model, on the basis of super-symmetric functional integrals, they obtain an effective $\sigma$-model approximation which, at the level of saddle point analysis, shows a localization/delocalization transition at $W \approx \sqrt{N}$.

Theorem \ref{thm:resolvent} is consistent with the above picture.  However, \cite{Casati:1990p3447, Casati:1993p35,Fyodorov:1991p5509} suggest that proper exponent on the r.h.s.\ of \eqref{eq:resolvloc} would be $\mu=2$.  
\begin{prob} What is the optimal value of $\mu$ in \eqref{eq:resolvloc}?  In particular, does this equation hold with $\mu =2$?
\end{prob}

In the physics literature, the nature of eigenvalue processes in the large $N$ limit is generally expected to be related to localization properties of the eigenfunctions, with Poisson statistics corresponding to localized eigenfunctions and Wigner-Dyson statistics corresponding to extended eigenfunctions. Let us call this idea the ``statistics/localization diagnostic.'' (In the context of band random matrices, a vector $\vec{v}$ is a function on the index set $\{1, \ldots,N\}$, namely $\vec{v}(i)=i^{\text{th}}$ coordinate of $\vec{v}$.  The statistics/localization diagnostic suggests that the eigenvalues of a random matrix should be approximately uncorrelated  if a typical eigenvector is  essentially supported on a vanishing fraction of $\{1, \ldots,N\}$, and should show strong correlations if it is typically  spread over more or less the entire index set.)

The extreme cases $W=1$ and $W=N$ of the Gaussian band ensemble \eqref{eq:GBE} are consistent the statistics/localization diagnostic. Indeed, with $W=1$, the matrix is diagonal and the eigenvalues, which are just the diagonal entries $d_{j,j}$, are independent.  After suitable rescaling the eigenvalue process converges to a Poisson process in the large $N$ limit.  (This is essentially the definition of a Poisson process.)  Likewise the eigenfunctions are  the elementary basis vectors $\vec{e}_{i}(j) =\delta_{i}(j)$, which are localized on single sites.  On the other hand, with $W=N$ the matrix $X_{W;N}$ is sampled from the GUE.  In this case, the eigenfunctions together form a uniformly distributed orthonormal frame, so they are completely extended, and a suitable rescaling of the eigenvalue process converges in distribution to an explicit determinental point process as calculated by Dyson \cite{Dyson:1962p6418, Dyson:1962p6420}.

Based on the statistics/localization diagnostic, it is reasonable to conjecture that Poisson statistics hold for local fluctuations of the eigenvalues of $X_{W;N}$ in a limit $N \ra \infty$ with $W=W(N) \ra \infty$ provided $W(N)^{\mu}/N \ra 0$. 
(One must be a little careful with the diagnostic, as it is  easy to concoct  random matrices with totally extended eigenfunctions and
arbitrary statistics: put $N$ random numbers with any given joint distribution on the diagonal of a matrix and conjugate the result with a random unitary! Of course, in that ensemble the matrix elements will most likely be highly correlated. Thus, it remains plausible that the statistics/localization diagnostic is correct, at least, for matrices with independent matrix elements.)

For random Schr\"odinger operators, Minami has derived Poisson statistics for the local correlations of the eigenvalue process from exponential decay of the resolvent \cite{Minami:1996fv}.  Some aspects of Minami's proof translate to the present context.  Most notably, the so-called \emph{Minami estimate} which bounds the probability of having two eigenvalues in a small interval,
\begin{equation}\label{eq:Minami}
\frac{1}{N^{2}}
\Pr\left [ \# \{ \lambda_{j} \in I \} \ge 2 \right ] \le C_{W} |I|^{2} ,
\end{equation}
where $|I|$ is the length of the interval and $\lambda_{1} \le \cdots \le \lambda_{N}$ are the eigenvalues of $X_{W;N}$,
holds with 
\begin{equation}
C_{W} \ \propto \ W^{2 \sigma} .
\end{equation}
Here $\sigma$ is as in Thm.\ \ref{thm:resolvent}.  (The proof of this fact may be accomplished by following Minami's argument or by one of the various alternatives that have appeared recently in the literature \cite{Graf:2007lp,Bellissard:2007yj, Combes:2008wf}.)

However, one crucial ingredient is missing: \emph{we lack sufficient control on the convergence of the density of states.}  The \emph{density of states} of $X_{W;N}$ is the measure $\kappa_{W;N}(\lambda) \di \lambda$ on the real line giving the density of the eigenvalue process:
\begin{equation}
\int_{I}\kappa_{W;N}(\lambda) \di \lambda \ = \  \frac{1}{N} \Ev{ \#  \{ \lambda_{j} \in I \}}.
\end{equation}
As indicated, $\kappa_{W;N}(\lambda) \di \lambda$ is absolutely continuous. In fact, it follows from the Wegner estimate \tem \ \eqref{eq:Weg1} below \tem \ that
\begin{equation}\label{eq:introWeg}
\abs{\kappa_{W;N}(\lambda)} \ \lesssim \ W^{\sigma},
\end{equation}
so analogous to \eqref{eq:Minami} we have
\begin{equation}
\frac{1}{N}
\Pr\left [ \# \{ \lambda_{j} \in I \} \ge 1 \right ] \le   \frac{1}{N} \Ev{ \#  \{ \lambda_{j} \in I \}} = \int_{I}\kappa_{W;N}(\lambda) \di \lambda \le W^{\sigma} |I|.
\end{equation}
(In fact, the Minami estimate is proved in a similar way, by showing that the \emph{expected number of eigenvalue pairs} in $I$ is bounded by the r.h.s. of \eqref{eq:Minami}.)

To study local fluctuations of the eigenvalue processes near $\lambda_{0} \in \R$, it is natural to consider the re-centered and re-scaled process
\begin{equation}\label{eq:rcrspr} \widetilde{\lambda}_{j} = N (\lambda_{j} - \lambda_{0}), \end{equation}
which has mean spacing $O(1)$.  We say that  \emph{the eigenvalue process has Poisson statistics near $\lambda_{0}$}, in some limit $W=W(N)$ and $ N \ra \infty$, if the point process $\{\widetilde{\lambda}_{j}\}$ converges to a Poisson process.   The density of this Poisson process would then be given by the limit $\lim_{N \ra \infty} \kappa_{W(N);N}(\lambda_{0})$.  The difficulty is \emph{we do not know that this limit exists.}

Now, for a fairly general class of matrix ensembles with
independent centered entries, e.g., for the Gaussian ensemble \eqref{eq:GBE}, it is known that the density of states $\kappa_{W;N}$ converges \emph{weakly} to  the semi-circle law, provided $W(N)/N \rightarrow 0$ or $1$ (see \cite{Molchanov:1992p5303}). That is, 
\begin{equation}\label{eq:ssl}
\frac{1}{N}  \Ev{\tr f(X_{W(N);N}) } \ = \ \int_{\R} f(\lambda) \kappa_{W;N}(\lambda) \di \lambda \ \xrightarrow[]{N \ra \infty}
\ \frac{1}{2\pi} \int_{-2}^2 f(t) \sqrt{4 -t^2} \di t .
\end{equation}
However, as indicated this is a weak convergence result, and it does not follow that
\begin{equation}\label{eq:strongconvergence}
\kappa_{W(N);N}(\lambda) \xrightarrow[]{N \ra \infty}
\ \frac{1}{2\pi} \sqrt{4 -\lambda^2} I[ |\lambda| \le 2],
\end{equation}
or even that
\begin{equation}\label{eq:shortintervalconvergence}
\int_{(\lambda_{0} - \frac{a}{N} , \lambda_{0} + \frac{b}{N}) } \kappa_{W;N}(\lambda) \di \lambda \ \xrightarrow[]{N \ra \infty}\frac{1}{2\pi} \sqrt{4 -\lambda^2} I[ |\lambda| \le 2],
\end{equation}
which would in fact be sufficient to control the density of the putative limit process.

In this regard, let us state a couple of open problems.  
\begin{prob} Improve the estimate \eqref{eq:introWeg}.  In particular, does this bound hold with $\sigma =0$? (The interpretation of $\kappa_{W;N}(\lambda)/N$ as the mean eigenvalue spacing and the convergence \eqref{eq:ssl} suggests that $\kappa$ should  be bounded.)
\end{prob}
\begin{prob} Verify either \eqref{eq:strongconvergence} or \eqref{eq:shortintervalconvergence}.  
\end{prob}

\subsection*{Acknowledgments} I would like to thank Tom Spencer and Michael Aizenman for many interesting discussions related to this and other works, and to express my gratitude for the hospitality extended me by the Institute for Advanced study, where I was member when this project started, and more recently by the Isaac Newton Institute during my stay associated with the program Mathematics and Physics of Anderson Localization: 50 Years After.

\section{Tridiagonal matrices}\label{sec:W=2}
The aim of this section is to motivate the assumptions on the distribution of $X_{W;N}$, spelled out below in \S\ref{sec:axioms}, by examining separately the somewhat simpler case  $W=2$.  Thus, consider for each $N \in \N$, 
a random tridiagonal matrix
\begin{equation}
X_{2;N} \ = \ \begin{pmatrix}v_{1} & t_{1}  \\
t_{1}^{*} & v_{2} & t_{2}  \\
& t_{2}^{*} & v_{3} & \ddots \\
&  & \ddots & \ddots & \ddots \\
& & & \ddots & v_{N-1} & t_{N-1} \\
& & & & t_{N-1}^{*} & v_{N}
\end{pmatrix},
\end{equation}
with $v_{1}, v_{2}, \ldots$ and $t_{1}, t_{2}, \ldots$ two given mutually independent sequences  of independent random variables, real and complex valued respectively.   For such matrices, exponential decay of the Green's function and localization of eigenfunctions can be obtained by the transfer matrix approach, see \cite{Carmona:1990ee}.  Here we use a different method, which is closely related to the technique of Kunz and Souillard \cite{Kunz:1980fd}.

As discussed above, the central technical estimate is a bound on $\field{E} (| \langle \vec{e}_{i}, (X_{2;N}-\lambda)^{-1} \vec{e}_{j} \rangle|^{s})$, decaying exponentially in the distance $|i-j|$. To obtain this bound, it is convenient to assume that $(v_k)$ are identically distributed and likewise $(t_k)$.  (This assumption could be replaced by uniformity in $k$ of certain bounds assumed below.  Likewise, strict independence of $(v_{k})$ is not really the issue. The argument could easily be adapted to the situation in which $(v_{k})$ are generated by a distribution with finite range coupling, such as $\prod_{k} \rho(v_{k}-v_{k-1}) \di v_{k}$.) The distribution of $(t_{k})$ can be arbitrary \tem \ theses variables may even be deterministic as in the case of random Jacobi matrices.  

To facilitate the fluctuation argument proposed above we will suppose the common distribution of $v_{k}$ has a density $\rho$ with the following property:
\begin{defn} We say that a probability density $\rho$ on $\R$ is \emph{fluctuation regular} if there are constants $\epsilon, \delta > 0$ and measurable set $\Omega \subset \R$ with $\int_{\Omega} \rho(v) \di v >0$ such that
\begin{equation}
v \in \Omega \ \implies \ \frac{\rho(v_{1})}{\rho(v_{2})} \ge \delta \quad \text{ for all } v_{1}, v_{2} \in (v-\epsilon, v+ \epsilon)
\end{equation}
\end{defn}
\begin{rem}A sufficient condition for fluctuation regularity is that $\ln \rho$ is Lipschitz on some open interval.  For example a uniform distribution $\rho(x) \propto \chi_{[a,b]}(x)$ is fluctuation regular.  So are the Gaussian and Cauchy distributions.  However,  fluctuation regularity is quite a bit stronger than absolute continuity of the measure $\rho \di x$, since it implies the existence of an \emph{open} set on which $\rho$ is strictly positive.  
\end{rem}
Our goal in this section is to prove the following result:
\begin{thm}\label{thm:W=2}Let $(v_{k})_{k=1}^{\infty}$ and $(t_{k})_{k=1}^{\infty}$ be two sequences of i.i.d. random variables, real and complex valued respectively.  Suppose that  the common distribution of $v_{k}$ has a density $\rho$ which is bounded and fluctuation regular.  Then for $0<s<1$ and $\Lambda > 0$ there are $A_{s} < \infty$ and $\mu_{s,\Lambda} >0$ such that for all $\lambda \in [-\Lambda,\Lambda]$,
\begin{equation}
\Ev{ |\langle \vec{e}_{i}, (X_{2;N}- \lambda)^{-1} \vec{e}_{j} \rangle|^{s} } \ \le \ A_{s}\e^{-\mu_{s,\Lambda} |i-j|}.
\end{equation}
\end{thm}
\begin{rem} We restrict $\lambda$ to a compact set to facilitate the fluctuation argument below.  In fact, for large $|\lambda|$ the rate of exponential decay will improve,  although the mechanism will be somewhat different.  One could construct a proof in this context along the lines of \cite{Aizenman:1994qm}.  Thus the $\Lambda$ dependence of the mass of decay $\mu_{s;\Lambda}$ may be dropped. 
\end{rem}

Let $g_{N}(i,j;\lambda) = \langle \vec{e}_{i}, (X_{2;N}-\lambda)^{-1} \vec{e}_{j} \rangle$.  Recall that  the decay of $\Ev{|g_{N}(i,j;\lambda)|^{s}}$ was to be established in two steps, the first of those being Lemma W which gives  finiteness of the fractional moments.   A preliminary observation is that Lemma W holds for these tridiagonal matrices:
\begin{lem}[Lemma W for $X_{2;N}$] \label{lem:wegW=2} Suppose that the distribution of $v_{k}$, $k=1,\ldots,N$ satisfies
\begin{equation}\Pr(v_{k} \in [a,b] ) \ \le \ \frac{\kappa}{2 \pi} |b-a| , 
\end{equation}
for any interval $[a,b]$, with $\kappa$ a finite constant.  Then
\begin{equation}\label{eq:gNbound}
\Pr (|g_{N}(i,j;\lambda) | > t | ( v_{k} )_{k \neq i, j}, \ (t_{k}) ) \ \le \ \frac{\kappa}{t},
\end{equation}
so, in particular,
\begin{equation}
\Pr (|g_{N}(i,j;\lambda)| > t ) \ \le \    \frac{\kappa}{t} 
\end{equation}
and
\begin{equation}
\Ev{ |g_{N}(i,j;\lambda)|^{s}} \ \le \ \frac{\kappa^{s}}{1-s}
\end{equation}
for $0 < s <1$.
\end{lem}
\begin{rem} The l.h.s.\ of \eqref{eq:gNbound} is the \emph{conditional probability} of the event $\{|g_{N}(i,j;\lambda)| >t\}$ at specified values of $( v_{k} )_{k \neq i, j}$ and $(t_{k}) $ \tem \ that is the  probability conditioned on the $\Sigma$ algebra generated by these variables. Eq.\ \eqref{eq:gNbound} is a standard estimate from the fractional moment analysis of discrete random Schr\"odinger operators, see \cite{Aizenman:1993p2946}.  The main point of this result is that to bound $\Ev{|g_{N}(i,j;\lambda)|^{s}}$, it is sufficient to average over $v_{i}$ and $v_{j}$.  
\end{rem}

The second part of the argument is to establish large fluctuations for $g_N(i,j;\lambda)$ \tem \ this is Lemma F above.  In the present context we have
\begin{lem}[Lemma F for $X_{2;N}$]\label{lem:fluctW=2} Under the hypotheses of Theorem \ref{thm:W=2}, for each $0 < r <s <1$ and $\Lambda \in \R$ there is  a constant $C_{r,s;\Lambda} < \infty$ such that
\begin{equation} \E\left(|g_{N}(i,j;\lambda)|
^{r}\right)\le\exp(-C_{r,s;\Lambda} |i-j|)\E\left(|g_{N}(i,j;\lambda)|
^{s}\right)^{r/s},\end{equation}
for $\lambda \in [-\Lambda, \Lambda]$. 
\end{lem}
\begin{rem} Together Lemmas \ref{lem:wegW=2} and \ref{lem:fluctW=2} prove Theorem \ref{thm:W=2}.
\end{rem}

\begin{proof} Let us fix $\lambda$ for the moment and drop it from the notation: $g_{N}(i,j) = g_{N}(i,j;\lambda)$.  Suppose  without loss of generality that $i < j$. A preliminary observation is that 
\begin{equation}\label{eq:w=2g}
g_{N}(i,j)  \  = \ - g_{j-1}(i, j-1) t_{j-1} g_{N}(j,j),
\end{equation}
which  may be established using the resolvent identity, writing $X_{2;N}$ as a perturbation of the corresponding matrix with $t_{j-1}$ set equal to zero (which decouples into two distinct blocks).  Iteration of this identity gives
\begin{equation}
g_{N}(i,j) \ = \ \left ( - 1\right )^{j-i} \left [ \prod_{k=i}^{j-1} g_{k}(k,k) t_{k} \right ] g_{N}(j,j) . 
\end{equation}
Thus
\begin{equation}
\ln| g_{N}(i,j)| \ = \  \sum_{k=i}^{j-1} \ln |t_{k}| + \sum_{k=i}^{j-1} \ln |g_{k}(k,k)|+ \ln |g_{N}(j,j)|  ,
\end{equation}
suggesting that if either $\ln |t_{k}|$ or $\ln |g_{k}(k,k)|$ were to exhibit fluctuations of order one, then the variance of $\ln |g_{N}(i,j)|$ would be of order $|i-j|$ and Lemma \ref{lem:fluctW=2} would follow. However, there are substantial correlations between the various terms, making it difficult to proceed directly along this line of argument.

To make a precise analysis, let us consider the random variables
\begin{equation}
\gamma_{k}= \frac{ 1}{g_{k}(k,k)} ,
\end{equation}
which are related by a recursion relation
\begin{equation}\label{eq:recurrence}
\gamma_{k} \ = \ v_{k} -  \lambda - \frac{|t_{k-1}|^{2}}{\gamma_{k-1}} , \quad 2 \le k \le N ,
\end{equation}
with
\begin{equation}
\gamma_{1} \ = \ v_{1}.
\end{equation}
These identities may be established using the Schur-complement formula.  In a similar way, the Schur-complement formula may be used to show that
\begin{equation}
\frac{1}{g_{N}(j,j) }= v_{j} - \lambda - \frac{|t_{j-1}|^{2}}{\gamma_{j-1}} - |t_{j}|^{2}  \widehat G_{j+1}
=\gamma_{j} - |t_{j}|^{2} \widehat G_{j+1}.
\end{equation}
where $\widehat G_{j+1} = \langle \vec{e}_{j+1}, (\widehat X_{2;N}- \lambda)^{-1} \vec{e}_{j+1}\rangle $ with $\widehat X_{2;N}$ the matrix obtained from $X_{2;N}$ by setting $t_{j}=0$:
\begin{equation}
\widehat X_{2;N} = \begin{pmatrix}
\ddots & \ddots\\
 \ddots & v_{j-1} & t_{j} \\
 & t_{j-1}^{*} & v_{j} & 0 \\
 & & 0 & v_{j+1} & t_{j+1} \\
 & & & t_{j+1}^{*} & v_{j+1} & \ddots \\
 & & & & \ddots & \ddots
\end{pmatrix}.
\end{equation}
In particular, $\widehat G_{j+1}$ \emph{is a function of the variables $(v_{k})_{k=j+1}^{N}$ and $(t_{k})_{k=j+1}^{N}$}.

We now make a change of variables $v_{k} \mapsto \gamma_{k}$ in our probability space. The Jacobian is triangular with ones on the diagonal and therefore has determinant one. Thus
\begin{equation}
\text{Joint distribution of $(\gamma_{k})_{k=1}^{N}$ given $(t_{k})_{k=1}^{N}$} \ = \ 
\rho(\gamma_{1} + \lambda)
\prod_{k=2}^N \rho(\gamma_k + \lambda + \frac{|t_{k-1}|^{2}}{\gamma_{k-1}} )  \, \prod_{k=1}^N \di \gamma_k.
\end{equation} 
So $\gamma_{k}$ are a chain of variables with nearest neighbor couplings \tem \ thinking of $k$ as a time parameter, $\{\gamma_{k}\}$ is a Markov chain.
In terms of these variables, we have
\begin{equation}
g_{N}(i,j) = (-1)^{|i-j|} \prod_{k=i}^{j-1} \frac{t_{k}}{\gamma_{k}} \times \frac{1}{\gamma_{j} - |t_{j}|^{2} \widehat G_{j+1}},
\end{equation}
where $\widehat G_{j+1}$ may be written as a function of $(\gamma_{k})_{k=j}^{N}$ and $(t_{k})_{k=j}^{N}$, since $v_{k} = \gamma_{k} +\lambda - |t_{k-1}|^{2}/\gamma_{k}$.

A useful trick for analyzing fluctuations in this context, inspired by the Dobrushin Shlosman analysis of continuous symmetries in $2D$ classical statistical mechanics \cite{Dobrushin:1975qx}, is to couple the system to a family of independent identically distributed random variables $\alpha_{2}, \alpha_{5}, \ldots$, each with absolutely continuous distribution $H(\alpha_{k}) \di \alpha_{k}$.  For technical reasons, which will become apparent below, we introduce $\alpha_{k}$ only for $k \equiv 2 \mod 3$.  Let us define 
\begin{equation}
f_k =  \e^{\alpha_k} \gamma_k ,
\end{equation}
where we take $\alpha_{k}=0$ for $k \not \equiv 0 \mod 3$. The Jacobian determinant of the transformation $(\gamma_{k},\alpha_{k}) \mapsto (f_{k}, \alpha_{k})$ is $\prod_{k=2,5,8,\ldots}^{N} \e^{-\alpha_k}$, so\begin{multline}
\text{joint distribution of $(f_k)_{k=1}^{N}$ and $(\alpha_k)_{k=1}^{N}$, given $(t_{k})_{k=1}^{N}$} \ = \\
\prod_{k\equiv 2 \mod 3}^{N} \rho(f_{k-1} + \lambda + \frac{|t_{k-2}|^{2}}{f_{k-2}}) \rho(\e^{-\alpha_{k}}f_{k} + \lambda + \frac{|t_{k-1}|^{2}}{f_{k-1}}) \rho(f_{k+1} + \lambda+ \e^{\alpha_{k}} \frac{|t_{k}|^{2}}{\gamma_{k}}) H(\alpha_{k}) \e^{- \alpha_{k}}  \\ \times \di f_{k-1} \di f_{k} \di f_{k+1} \di \alpha_{k},
\end{multline}
with the convention that $t_{0}=0$. 

We now \emph{fix} $(f_{k})_{k=1}^{N}$, and consider  the conditional distribution of $(\alpha_{k})_{k=1}^{N}$, which carries some information on the distribution of $(\gamma_{k})_{k=1}^{N}$. A key point is that the variables $\alpha_k$ \emph{remain independent after conditioning}.  They are, however, no longer identically distributed. Instead, 
\begin{multline}
\text{distribution of $\alpha_{k}$ given $(t_{\ell})_{\ell=1}^{N}$ and $(f_{\ell})_{\ell=1}^{N}$} \\
= \frac{ \rho(\e^{-\alpha_k} f_k + \lambda + \frac{|t_{k-1}|^{2}}{f_{k-1}} ) \rho(f_{k+1} + \lambda +  \e^{\alpha_k}  \frac{|t_{k}|^{2}}{f_{k}}) H(\alpha_k) \e^{-\alpha_k} }{Z_k} \di \alpha_{k}
\end{multline}
with 
\begin{equation}
Z_k \ = \ \int \rho(\e^{-\alpha} f_k +  \lambda + \frac{|t_{k-1}|^{2}}{f_{k-1}} ) \rho(f_{k+1} + \e^{\alpha}  \frac{|t_{k}|^{2}}{f_{k}}) H(\alpha) \e^{-\alpha} \di \alpha .
\end{equation}

We now express $g_{N}(i,j)$ in terms of the variables $(t_{\ell}, f_{\ell},\alpha_{\ell})_{\ell=1}^{N}$,\begin{equation}g_{N}(i,j) \ = \ (-1)^{j-i} \Biggl [ \prod_{\substack{k \equiv 2 \mod 3  \\ i \le k \le j-1}} \e^{\alpha_{k}} \Biggr ] \  \Biggl [  \prod_{k=i}^{j-1} \frac{t_{k}}{f_{k}} \   \Biggr ]  \  \widehat H_{j+1}, 
\end{equation}
where
\begin{equation}
\widehat H_{j+1} = 
\frac{1}{\e^{-\alpha_{j}}f_{j} - |t_{j}|^{2} \widehat G_{j+1}}
\end{equation}
is a function of $(t_{\ell},f_{\ell},\alpha_{\ell})_{\ell=j}^{N}$. By the conditional independence of $(\alpha_{k})$ we find that 
\begin{multline}\Ev{|g_{N}(i,j)|^{r} \left | (t_{\ell},f_{\ell})_{\ell=1}^{N}, \ (\alpha_{\ell})_{\ell=j}^{N}
\right .} \\ = \ \ \left ( \prod_{k=i}^{j-1}\frac{|t_{k}|^{r}}{|f_{k}|^{r}}  \right )   \widehat H_{j+1}  \ \prod_{\substack{k \equiv 2 \mod 3  \\ i \le k \le j-1}} \Ev{ \e^{r\alpha_{k}}\left | (t_{\ell},f_{\ell})_{\ell=1}^{N} \right .} .
\end{multline}
Applying propostion \ref{prop:basic} to each factor $\Ev{ \e^{r\alpha_{k}} | (t_{\ell},f_{\ell})_{\ell=1}^{N}} $ on the right hand side, we find that
\begin{multline}
\Ev{|g_{N}(i,j)|^{r} \left | (t_{\ell},f_{\ell})_{\ell=1}^{N}, \ (\alpha_{\ell})_{\ell=j}^{N}
\right .} \\ = \ \left ( \prod_{k=i}^{j-1}\frac{|t_{k}|^{r}}{|f_{k}|^{r}}  \right )   \widehat H_{j+1} \ \prod_{\substack{k \equiv 2 \mod 3  \\ i \le k \le j-1}} \e^{-h_{k}(r,s)} \Ev{\e^{s\alpha_{k}}\left | (t_{\ell},f_{\ell})_{\ell=1}^{N}
\right .}^{r/s} ,
\end{multline}
with 
\begin{equation}\label{eq:hintegral}
h_{k}(r,s) = \frac{1}{s} \int_{0}^{s} \min(r,q) (s- \max(r,q)) \Var_{q}(\alpha_{k} |  (t_{\ell},f_{\ell})_{\ell=1}^{N}) \di q ,
\end{equation}
and 
\begin{equation}\label{eq:conditionalvariance}
\Var_{q} (\alpha_{k} | (t_{\ell},f_{\ell})_{\ell=1}^{N}) \ = \ \inf_{m \in \R} \frac{\Ev{(\alpha_{k} -m)^{2}\e^{q \alpha_{k}}|(t_{\ell},f_{\ell})_{\ell=1}^{N}}}{\Ev{\e^{q \alpha_{k}}|(t_{\ell},f_{\ell})_{\ell=1}^{N}}} .
\end{equation}
Using the conditional independence of $(\alpha_{k})$ once again to reassemble $g_{N}(i,j)$ inside the expectation on the r.h.s., we find that
\begin{multline}
\Ev{|g_{N}(i,j)|^{r} | (t_{\ell},f_{\ell})_{\ell=1}^{N}, \ (\alpha_{\ell})_{\ell=j}^{N}} \\ = \ \e^{- \sum_{k=i}^{j-1} h_{k}(r,s)}  \, \Ev{|g_{N}(i,j)|^{s} | (t_{\ell},f_{\ell})_{\ell=1}^{N}, \ (\alpha_{\ell})_{\ell=j}^{N}} ^{r/s} ,
\end{multline}
where we have set $h_{k}(r,s)=0$ for $k \not \equiv 2 \mod 3$.  

After averaging and applying the H\"older inequality, we conclude that
\begin{equation}
\Ev{|g_{N}(i,j)|^{r}} \ \le \ \Ev{\e^{- \frac{s}{s-r} \sum_{k=i}^{j-1} h_{k}(r,s)}}^{\frac{s-r}{s}}  \Ev{ |g_{N}(i,j)|^{s}}^{r/s}.  \label{eq:keyresult}
\end{equation}
Eq.\ \eqref{eq:keyresult} is the key result. The exponent in the first factor is a sum of $O(N)$ non-negative terms, each presumably $O(1)$ and positive with positive probability.  It will not be so surprising to find that the term itself is $O(N)$ with good probability.  The rest is estimates.

To proceed with the estimates, let us take the \emph{a priori} distribution of $\alpha_{k}$, before coupling and conditioning, to be uniform in an interval $[-\lf,\lf]$ centered at the origin:
\begin{equation}
H(\alpha) \ = \ \frac{1}{2 \lf} I[|\alpha| < \lf],
\end{equation}
with $\lf$ to be chosen below.
Although $ \Var_{q} (\alpha_{k} | (t_{\ell},f_{\ell})_{\ell=1}^{N})$ is defined as a function of $ (t_{\ell},f_{\ell})_{\ell=1}^{N}$, it is useful to express it in terms of the variables $(t_{\ell}, \gamma_{\ell}, \alpha_{\ell})_{\ell=1}^{N}$:
\begin{equation}
 \Var_{q} (\alpha_{k} |  (t_{\ell},f_{\ell})_{\ell=1}^{N}) \ = \ \inf_{m \in \R} \frac{\int_{-\lf}^{\lf} (\alpha- m)^{2} \e^{(q-1)\alpha}\nu_{k}(\alpha) \di \alpha}{\int_{-\lf}^{\lf} \e^{(q-1)\alpha} \nu_{k}(\alpha) \di \alpha}
\end{equation}
with 
\begin{equation}\label{eq:nuk}
\nu_{k} (\alpha) \ = \ \rho(\e^{\alpha_{k} - \alpha} \gamma_{k} + \lambda +\frac{|t_{k-1}|^{2}}{\gamma_{k-1}} ) \rho(\gamma_{k+1} + \lambda +\e^{\alpha - \alpha_{k}} \frac{|t_{k}|^{2}}{\gamma_{k}}) .
\end{equation}

A lower bound for $ \Var_{q} (\alpha_{k} |  (t_{\ell},f_{\ell})_{\ell=1}^{N})$, sufficient for our purposes, is 
\begin{equation}
 \Var_{q} (\alpha_{k} | (t_{\ell},f_{\ell})_{\ell=1}^{N})\ \ge \ \frac{1}{3} \e^{-2|q-1|\lf}  \lf^{2} \, \frac{\inf_{-\lf < \alpha <\lf} \nu_{k}(\alpha)}{\sup_{-\lf< \alpha < \lf}  \nu_{k}(\alpha)} .\end{equation}
The r.h.s.\ still carries some dependence on $\alpha_{k}$, through the density $\nu_{k}$.  We may eliminate the dependence on $\alpha_{k}$ entirely by bounding the right hand side from below: 
\begin{multline}
 \Var_{q} (\alpha_{k} |  (t_{\ell},f_{\ell})_{\ell=1}^{N}) \\ \ge \ \frac{1}{3} \lf^{2} \e^{-2|q-1|\lf} \inf_{-2\lf < \alpha, \beta < 2\lf} \frac{ \rho(\e^{-\alpha}  \gamma_{k}+  \lambda +  \frac{|t_{k-1}|^{2}}{\gamma_{k-1}} ) \rho(\gamma_{k+1} + \lambda + \e^{\alpha} \frac{|t_{k}|^{2}}{\gamma_{k}} )}{ \rho(\e^{-\beta}  \gamma_{k}+  \lambda +  \frac{|t_{k-1}|^{2}}{\gamma_{k-1}} ) \rho(\gamma_{k+1} + \lambda + \e^{\beta} \frac{|t_{k}|^{2}}{\gamma_{k}} )} .
\end{multline}
It is useful to write
$$ \e^{-\alpha}  \gamma_{k}+\lambda + \frac{|t_{k-1}|^{2}}{\gamma_{k-1}}  \ = \ (\e^{-\alpha} -1) \gamma_{k} + v_{k},$$
and similarly for the term in the denominator and the term with index $k+1$.  Finally, the r.h.s.\ is no larger if we factor the infimum on the right hand side,
\begin{multline}\label{eq:factoredlowerbound}
 \Var_{q} (\alpha_{k} |  (t_{\ell},f_{\ell})_{\ell=1}^{N}) \\ \ge \ \frac{1}{3} \lf^{2} \e^{-2|q-1|\lf}  \inf_{-2\lf < \alpha, \beta < 2\lf} \frac{ \rho(v_{k} +(\e^{-\alpha} -1) \gamma_{k}) }{ \rho(v_{k} + (\e^{-\beta}-1)  \gamma_{k})} \inf_{-2\lf < \alpha, \beta < 2\lf}\frac{\rho(v_{k+1} + (\e^{\alpha}-1)  \frac{|t_{k}|^{2}}{\gamma_{k}}) }{ \rho(v_{k+1} + (\e^{\beta}-1) \frac{|t_{k}|^{2}}{\gamma_{k}})} .
\end{multline}

On the r.h.s., the only dependence on $q$ is in the exponential term.  In the integral \eqref{eq:hintegral}, there is not much loss in replacing this exponential by the (smaller) $\e^{-2|s-1|\lf}$, so that
\begin{equation}
\frac{s}{s-r} h_k(r,s) \ \ge \ \frac{r s}{6} \lf^2 \e^{-2 |s-1| \lf} U_k(\lf) ,
\end{equation}
with
\begin{equation}
U_k(\lf) = \inf_{-2\lf < \alpha, \beta < 2\lf} \frac{ \rho(v_{k} +(\e^{-\alpha} -1) \gamma_{k}) }{ \rho(v_{k} + (\e^{-\beta}-1)  \gamma_{k})} \inf_{-2\lf < \alpha, \beta < 2\lf}\frac{\rho(v_{k+1} + (\e^{\alpha}-1)  \frac{|t_{k}|^{2}}{\gamma_{k}}) }{ \rho(v_{k+1} + (\e^{\beta}-1) \frac{|t_{k}|^{2}}{\gamma_{k}})} .
\end{equation}
Plugging this estimate into eq.~\eqref{eq:keyresult}, we obtain \begin{equation}\label{eq:keyresultredux}
\Ev{|g_{N}(i,j)|^{s}} \ \le \ \Ev{\e^{-  \frac{rs}{6} \lf^2 \e^{-2|s-1|\lf} \sum_{k=i}^{j-1} U_k(\lf)}}^{\frac{s-r}{s}}  \Ev{ |g_{N}(i,j)|^{s}}^{r/s}.
\end{equation}

Since $\rho$ is fluctuation regular, there are $\delta, \epsilon >0$ and a set $\Omega \subset \R$ with $$\int_{\Omega} \rho \di x = q_{0} >0$$ such that $U_{k}(\lf) \ge \delta^{2} I[A_{k}] $ where $I[A_{k}]$ is the indicator function of the event:
\begin{equation}
A_{k} = \set{v_{k}, v_{k+1} \in \Omega\ , \quad |\gamma_{k}| \le \frac{\epsilon}{ \e^{2\lf} -1}, \quad \text{and} \quad \frac{|t_{k}|^{2}}{|\gamma_{k}|} \le \frac{\epsilon}{ \e^{2\lf} -1}}.
\end{equation}
In turn, since $\gamma_{k} = v_{k} +  \lambda + |t_{k-1}|^{2}/\gamma_{k-1}$ and $|\lambda| \le \Lambda$ (by assumption), we see that
\begin{multline}\label{eq:Aksubset}
A_{k} \supset \set{v_{k} \in \Omega \ , \ |v_{k}| \le L} \cap \set{v_{k+1} \in \Omega} \cap \set{ |t_{k-1}|, |t_{k}| \le \tau} \\ \cap \set{\frac{1}{|\gamma_{k-1}|} \le \frac{1}{\tau^{2}} \left ( 
 \frac{\epsilon}{ \e^{2\lf} -1} - L -  \Lambda \right ) } \cap \set{\frac{1}{|\gamma_{k}|} \le \frac{1}{\tau^{2}}  \frac{\epsilon}{ \e^{2\lf} -1}},\end{multline}
with $\tau$ and $L$ any positive numbers.

We estimate the probability of $A_{k}$ from below by integrating \eq{eq:Aksubset} over $v_{k+1}$, $v_{k}$, $v_{k-1}$, $t_{k}$, and $t_{k-1}$ in that order. (The need to integrate over three consecutive $v$ variables is the reason we introduced $\alpha_{k}$ only for $k \equiv 2 \mod 3$.)  To begin,
\begin{equation}
\Pr(v_{k+1} \in \Omega | (v_{l})_{l \neq k} , \ (t_{l}) ) = \int_{\Omega} \rho(v) \di v = q_{0}.
\end{equation}
Looking now at $v_{k}$, since $\gamma_{k} = v_{k} +  \lambda + |t_{k-1}|^{2}/\gamma_{k-1}$, we see that
\begin{multline}
\set{v_{k} \in \Omega \ , \ |v_{k}| \le L} \cap \set{\frac{1}{|\gamma_{k}|} \le \frac{1}{\tau^{2}}  \frac{\epsilon}{ \e^{2\lf} -1}} \\ =
\set{v_{k} \in \Omega} \cap \set{ |v_{k}| \le L} \cap \set{v_{k}  \not \in [a - \tau^{2 }\frac{\e^{2 \lf}-1}{\epsilon}, a + \tau^{2}\frac{\e^{2\lf}- 1}{\epsilon}]}
\end{multline}
with $a = \lambda + |t_{k-1}|^{2}/\gamma_{k-1}$.  Since the density $\rho$ is bounded, it follows that
\begin{multline}
\Pr \left ( \left. v_{k} \in \Omega \ , \ |v_{k}| \le L\ , \ \frac{1}{|\gamma_{k}|} \le \frac{1}{\tau^{2}}  \frac{\epsilon}{ \e^{2\lf} -1} \right | (v_{l})_{l \neq k,k+1} , \ (t_{l}) \right ) \\ \ge q_{0} - \Pr(|v_{k}| > L) -2 \norm{\rho}_{\infty} \tau^{2}  \frac{\e^{2\lf}- 1}{\epsilon}.
\end{multline}
Similarly
\begin{multline}
\Pr \left (\left . \frac{1}{|\gamma_{k-1}|} \le \frac{1}{\tau^{2}} \left ( 
 \frac{\epsilon}{ \e^{2\lf} -1} - L -\Lambda \right ) \right | (v_{l})_{l \neq k-1,k,k+1}, \ (t_{l}) \right ) \\ 
 \ge 1 -  2 \norm{\rho}_{\infty} \tau^{2} \frac{1}{\frac{\epsilon}{ \e^{2\lf} -1} - L - \Lambda}.
\end{multline}
Combining these estimates with \eq{eq:Aksubset} and integrating over the identically distributed variables $t_{k}$ and $t_{k-1}$, we find
\begin{multline}\label{eq:Akest}
\Pr(A_{k} | (v_{l})_{l \neq k-1,k,k+1} , \ (t_{l})_{l \ne k, k-1}) \ \ge \ q_{0} \left ( q_{0} - \Pr(|v_{k}| > L) -2 \norm{\rho}_{\infty} \tau^{2}  \frac{\e^{2\lf}- 1}{\epsilon} \right )\\ \times \left ( 1 -  2 \rho_{\infty} \tau^{2} \frac{1}{\frac{\epsilon}{ \e^{2\lf} -1} - L - \Lambda} \right ) \Pr(|t_{k}| \le \tau )^{2} .
\end{multline}
The key things to observe is that  the r.h.s.\ of \eq{eq:Akest} is independent of $k$ and can be made arbitrarily close to $q_{0}^{2}$ by suitable choice of large $L$, $\tau$ and small $\eta$. 

So, for sufficiently small $\eta$ we have $\Pr(A_{k} | (v_{l})_{l \neq k-1,k,k+1} , \ (t_{l})_{l \ne k, k-1})) \ge \half q_{0}^{2}$, say.  Since $U_{k} (\eta) \ge \delta^{2} I[A_{k}]$, we find that 
\begin{equation}
\Ev{\e^{-  \frac{rs}{6} \lf^2 \e^{-|s-1|\lf} \sum_{k=i}^{j-1} U_k(\lf)}}^{\frac{s-r}{s} }
\ \le \ \exp \left ( - \frac{s-r}{2s} q_{0}^{2}  \left (1 - \e^{- \delta^{2} \lf^2 \frac{rs}{6}  \e^{-2|s-1|\lf}} \right )  \left \lfloor \frac{|i-j|}{3}  \right \rfloor\right ),
\end{equation}
by  integrating successively over $v_{k}, t_{k}$ from $k=i, \ldots, j-1$ (see Lemma \ref{lem:prob} below).
Combined with \eqref{eq:keyresultredux} this completes the proof of Lemma \ref{lem:fluctW=2}.   \end{proof}

\section{Band matrices}\label{sec:axioms}
To translate the argument of the previous section to the context of band matrices, we replace each of the variables $v_{j}$ and $t_{j}$ by $W \times W$ matrices.  Given $W \in \N$, consider a sequence, $V_{j}$, $n=1,\ldots$, of independent identically distributed hermitian $W \times W$ matrices together with a sequence, $T_{j}$, $n=1,\ldots$, of independent identically distributed $W \times W$ matrices (not necessarily hermitian).  With these matrix variables, we form an infinite random hermitian band matrix
\begin{equation}\label{eq:bandblockform}
X_{W} \ = \ \left(\begin{array}{cccccc} V_{1} & T_{1} & 0 \\
T_{1}^{\dagger} & V_{2} & T_{2} & \ddots\\
 0 & T_{2}^{\dagger} & V_{3} & &  \ddots  \\
 & \ddots & & \ddots &   \\
 & &  \ddots & & \ddots \end{array}\right),
 \end{equation}
 a random operator on $\ell^{2}(\N)$, 
and for each $N$ the random matrix
\begin{equation}\label{eq:XWN}
X_{W;N} \ = \ Q_{N} X_{W} Q_{N}
\end{equation}
with $Q_{N}$ the projection onto $\ell^{2}(\{1, \ldots, N\}).$   For simplicity,  let us consider only $N$ a multiple of $W$: $N = nW$.  Thus, 
\begin{equation}\label{eq:blockbandform2}
X_{W;N} \ = \ X_{W;nW} \ = \  \left(\begin{array}{cccccc} V_{1} & T_{1} & 0 \\
T_{1}^{\dagger} & V_{2} & T_{2} & \ddots\\
 0 & T_{2}^{\dagger} & V_{3} & &  \ddots  \\
 & \ddots & & \ddots & & 0 \\
& & \ddots & & \ddots &  T_{n} \\
 & & &  0 & T_{n-1}^{\dagger} & V_{n} \end{array}\right),
 \end{equation}
 Let $P_{j}$ denote the projection onto the $j^{\text{th}}$ block, $\ell^{2}(\{(j-1)W +1, \ldots, jW\})$, so
 $$V_{j} = P_{j} X_{W} P_{j} \quad \text{and} \quad T_{j} = P_{j} X_{W}P_{j+1}.$$
 
 Band matrix ensembles such as the Gaussian band ensemble \eqref{eq:GBE} are of this form,  with $T_{j}$ lower triangular matrices. However, for the argument presented below it is not necessary that $T_{j}$ be lower triangular. (Also, neither strict independence nor identicality of distribution are needed. Nonetheless, to  keep things simple, let us stick to the i.i.d.\ case.)

In adapting the arguments from the scalar case to the matrix variables  $V_{j},  T_{j}$, we must account for the non-commutativity of the matrix product. The basis of the argument is a change of variables $V_{j} \mapsto \e^{\alpha_{j}} \Gamma_{j}$ with $\alpha_{j}$ a scalar random variable and $\Gamma_{j}$ a $W \times W$ matrix obtained from the resolvent of $X_{W;jW}$.  In the end we will need to estimate the ratio
$$ \frac{\rho(V_{j } + (\e^{-\alpha}-1) \Gamma_{j})}{\rho(V_{j} + (\e^{-\beta}-1) \Gamma_{j})}$$
for small $\alpha$, $\beta$,  where $\rho$ is the density of the distribution of $V_{j}$ (assumed to be absolutely continuous with respect to Lebesgue measure on some vector space of matrices). In the scalar case, this change of variables was useful for all  fluctuation regular densities.  In the matrix case, an additional complication arises.  Unless $\Gamma_{j}$ falls in the vector space supporting the distribution of $V_{j}$ there will be constraints on the matrix elements of $\Gamma_{j}$ which manifest themselves as $\delta$ functions after the change of variables.   However, $\Gamma_{j}$ is formed from $\{ V_{k}\}$ and $\{T_{k} \}$ via \emph{non-linear} operations, so \emph{there is no reason to expect it to fall in this vector space}.  (For example when $V_{j}$ are diagonal, $\Gamma_{j}$ will in general have off-diagonal components.)  To guarantee closure under non-linear operations we suppose that the vector space supporting the distribution of $V_{j}$ is a  \emph{matrix algebra}:
\begin{defn}A \emph{$\star$ algebra over $\R$ of $W\times W$ matrices} is a set $\cu{A}$ of $W \times W$ matrices that is a vector space over $\R$, under the usual addition and scalar multiplication, and such that $$ V_{1},V_{2} \in \cu{A} \ \implies \ V_{1} V_{2} \in \cu{A} \quad \text{and} \quad V_{1}^{\dagger} \in \cu{A}.$$\end{defn}
We will use 
\begin{prop}\label{prop:cstar}If $\cu{A}$ is a matrix $\star$ algebra over $\R$ and $V \in \cu{A}$ is invertible then $V^{-1} \in \cu{A}$.
\end{prop}
\begin{proof}This is a standard result for $C^{\star}$ algebras.  In that context, the algebra is usually assumed to be a vector space over $\C$, but that is not necessary.  Here is the proof. If $V \in \cu{A}$ is self-adjoint and invertible, by the Weierstrass theorem we can approximate $V^{-1}$ (in the operator norm, say) by polynomials in $V$ \emph{with real coefficients}.  That is, we can approximate $V^{-1}$ by elements of $\cu{A}$.  Since a finite dimensional vector space is complete, $V^{-1} \in \cu{A}$.  For general invertible $V \in \cu{A}$, we have 
$ V^{-1} = (V^{\dagger} V)^{-1} V^{\dagger} \in \cu{A},$
since $V^{\dagger}V \in \cu{A}$ is self adjoint.
\end{proof}

\begin{ass}\label{ass:0}
Let $\cu{S}$ be an increasing sequence of integers and fix, for each $W \in \cu{S}$, a $\star$ algebra over $\R$ of $W\times W$ matrices $\cu{A}_{W}$, and the set $\cu{T}_{W}$ of matrices which preserve  $\cu{A}_{W}$ under conjugations
\begin{equation}
\cu{T}_{W}  \ = \ \set{ T \ : \  T^{\dagger} \cu{A}_{W} T \subset \cu{A}_{W}}.
\end{equation}
Let $\cu{A}_{W}^{H} =  \set{V \in \cu{A}_{W} \ : \ V = V^{\dagger} },$
the set of hermitian elements of $\cu{A}_{W}$.  \emph{We require that $T_{j} \in \cu{T}_{W}$ and $V_{j} \in \cu{A}_{W}^{H}$, $j=1, \ldots.$}
\end{ass}
\begin{rem} Note that $\cu{T}_{W}$ is closed under conjugation: $T \in \cu{T}_{W} \implies T^{\dagger} \in \cu{T}_{W}.$  
\end{rem}

 There is a good deal of flexibility in the choice of algebras.  Of course, we may take $\cu{A}_{W} = \cu{T}_{W} = $ all $n\times n$ complex matrices, so $X_{W;N}$ is complex Hermitian.  On the other hand, we could restrict $\cu{A}_{W}$ to be the set of matrices with real entries, so $X_{W;N}$ is real symmetric.  In this case $\cu{A}_{W}$ is \emph{not} a complex vector space. Similarly we could take $\cu{A}_{W}$ to be the set of matrices with quaternion entries, where the quaternions units are represented by $2 \times 2$ matrices, so $X_{W;N}$ would by Hermitian but anti-symmetric under transposition $X_{W;N}^{T} = - X_{W;N}$. In this last case, $\cu{S}$ would be the set of even integers.

An important consequence of assuming that $T_{j} \in \cu{T}_{W}$ and $V_{j} \in \cu{A}_{W}^{H}$,  is that we have some \emph{a priori} information on the block matrices making up the resolvent of $X_{W;nW}$.  \begin{lem}\label{lem:algebra} 
Suppose $Y$ is an $nW \times nW$ matrix that is block tri-diagonal,
$$P_{i} Y_{n} P_{j} = 0 \quad \text{if } |i-j| \ge 2,$$
and satisfies 
$$V_{j} = P_{j} Y P_{j} \in \cu{A}_{W}, \quad j=1, \ldots, n$$
and
$$T_{j} = P_{j} YP_{j+1}  = (P_{j+1} Y P_{j})^{\dagger} \in \cu{T}_{W},
\quad j=1, \ldots, n-1.$$
If $Y$ is invertible then
\begin{equation}\label{eq:algebra1}
P_{j} Y^{-1} P_{j} \in \cu{A}_{W}, \quad j=1,\ldots,n
\end{equation}
and
\begin{equation}\label{eq:algebra2}
P_{i} Y^{-1} P_{j} \in \cu{AT}_{W} , \quad i,j=1, \ldots, n,
\end{equation}
where $\cu{AT}_{W}$ is the algebra generated by $\cu{A}_{W}$ and $\cu{T}_{W}$.
\end{lem}
\begin{rem} The off diagonal blocks of $Y^{-1}$ need not be in $\cu{A}_{W}$.  This is apparent already for $n=2$, where, by the Schur complement formula,
$$ \begin{aligned} P_{1} Y^{-1} P_{2} &= ( V_{1} - T_{1} V_{2}^{-1} T_{1}^{\dagger} )^{-1} T_{1} V_{2}^{-1} \\&= V_{1}^{-1} T_{1} (V_{2} - T_{1}^{\dagger} V_{1}^{-1} T_{1} )^{-1}.
\end{aligned}$$
In each expression on the right, the first and last factors are in $\cu{A}_{W}$ but the middle factor, $T_{1}$, is not. 
\end{rem}
\begin{proof} The proof is by induction on $n$.  The result is clear for $n=1$.  So, suppose we know that it holds if $Y$ is a tridiagonal block matrix of size no larger than $(n-1)W \times (n-1)W$. 

First consider \eqref{eq:algebra1}. By the Schur complement formula,
$$ P_{j} Y^{-1} P_{j} = (V_{j} - T_{j} P_{j+1}Y_{+}^{-1} P_{j+1} T_{j}^{\dagger} - T_{j-1}^{\dagger} P_{j-1}Y_{-}^{-1} P_{j-1}T_{j-1})^{-1},$$
where
$$Y_{+} = \begin{pmatrix} V_{j+1} & T_{j+1}  \\
T_{j+1}^{\dagger} & \ddots & & \ddots \\
& \ddots & & \ddots &  T_{n-1} \\
&& & T_{n-1}^{\dagger} & V_{n}
\end{pmatrix}, \quad Y_{-} = \begin{pmatrix} V_{1} & T_{1}  \\
T_{1}^{\dagger} & \ddots & & \ddots \\
& \ddots & & \ddots &  T_{j-1} \\
&& & T_{j-1}^{\dagger} & V_{j}
\end{pmatrix}.
$$
As $Y_{+}$ and $Y_{-}$ are of size no larger than $(n-1) W \times (n-1)W$ and $T_{j}, T_{j+1} \in \cu{T}_{W}$,  it follows that 
$$ T_{j} P_{j+1}Y_{+}^{-1} P_{j+1} T_{j}^{\dagger},  \ T_{j-1}^{\dagger} P_{j-1}Y_{-} P_{j-1}T_{j-1} \in \cu{A}_{W}.$$
By Prop.\ \ref{prop:cstar} $P_{j} Y^{-1} P_{j} \in \cu{A}_{W}.$

Now consider \eqref{eq:algebra2}. Suppose $i < j$ (the other case is similar). 
Let 
$$ \widehat Y = \begin{pmatrix} V_{1} & T_{1}  \\
T_{1}^{\dagger} & \ddots & & \ddots \\
& \ddots & & \ddots &  T_{j-2} \\
&& & T_{j-2}^{\dagger} & V_{j-1}
\end{pmatrix}.
$$
By the resolvent identity, one has
$$ P_{i} Y^{-1} P_{j} =  - P_{i} \widehat Y^{-1} P_{j-1} T_{j} P_{j} Y^{-1} P_{j}.$$
But $P_{i} \widehat Y^{-1} P_{j-1} \in \cu{AT}_{W}$ by the induction hypothesis and $P_{j} Y^{-1} P_{j} \in \cu{A}_{W}$ as we have just shown.  It follows that the r.h.s.\ is in $\cu{AT}_{W}.$
\end{proof}

We now consider the properties required of the distribution of $V_{j}$, denoted $\field{P}_{W}$. Let $\norm{ \cdot }$ denote the operator norm of a matrix
\begin{equation}
\norm{A} = \sup_{\norm{\vec{v}}=1 } \norm{A \vec{v}}
\end{equation}
and let $\sigma(A)$ denote the set of eigenvalues of a matrix.  Recall, if $A$ is self-adjoint, that 
$$\norm{A} = \max |\sigma(A)| \quad \norm{A^{-1}} = \frac{1}{\min | \sigma(A)|}.$$
 \begin{ass}\label{ass:1}
Let $(\field{P}_{W})_{W\in \cu{S}}$ be a family of probability measures such that
\begin{itemize}
\item  (\emph{Absolute continuity}): Each measure $\field{P}_{W}$ is supported on $\cu{A}_{W}^{H}$ and  absolutely continuous with respect to Lebesgue measure on that space.  Let $\rho_{W}(V)$ denote the density of $\field{P}_{W}$ with respect to Lebesgue measure.
\item (\emph{Wegner-type estimates}): There are $\kappa > 0$ and $\sigma \ge 0$ such that for all $A \in \cu{A}_{W}^{H}$, $W \in \cu{S}$,
\begin{equation}\label{eq:Weg1}\field{P}_{W}\set{V \ : \ \norm{(V-A)^{-1}} > tW^{1+\sigma} } \ \le \  \kappa  \frac{1}{t}; \end{equation}
and for all $A, B \in \cu{A}_{W}^{H}$ and $C \in \cu{AT}_{W}$, $W \in \cu{S}$,
\begin{equation}\label{eq:Weg2}\field{P}_{W} \otimes \field{P}_{W} \set{(V_{1}, V_{2}) \ : \ 
\norm{ \begin{pmatrix}
V_{1} - A & C \\
C^{\dagger} & V_{2}- B
\end{pmatrix}^{-1} } > tW^{1+\sigma} } \ \le \ 2 \kappa \, \frac{1}{t}.\end{equation}

\item (\emph{Fluctuation regularity with bounded tails}): There are constants $p_{0}, \delta, \epsilon >0$, $L, a, \zeta \ge 0$  such that, for each $W \in \cu{S}$, there is $\Omega_{W} \subset \cu{A}_{W}^{H}$ with $\bb{P}_{W}(\Omega_{W}) \ge p_{0}$ and if $V \in \Omega_{W}$, then
\begin{equation}\label{eq:uniformnorma} \norm{V} \le L W^{a} 
\end{equation}
and
\begin{equation}\label{eq:Wfluctreg}  \frac{\rho_{W}(V_{1})}{\rho_{W}(V_{2} )} \ge \delta    \text{ for all } V_{1}, V_{2} \in \cu{A}_{W}^{h}  \text{ with } \norm{V_{j}-V} \le \epsilon W^{-\zeta}, \ j=1,2.
\end{equation}
\end{itemize}
\end{ass}
\begin{rems} ~

\begin{enumerate} 
\item Since
$$\norm{(V- \lambda I)^{-1}} = \frac{1}{\dist(\lambda, \sigma(V))},$$
the Wegner-type estimate \eqref{eq:Weg1} implies 
\begin{equation}\label{eq:fr}\field{P}_{W} \set{V \ : \ \dist(\lambda, \sigma(V)) \le\frac{\epsilon}{W^{1+ \sigma}}} \ \le \ \kappa \epsilon.\end{equation}
If $V$ is suitably scaled so as to have mean eigenvalue spacing of order $1/W$,  this suggests that  we should be able to take $\sigma =0$.  That has not been proved, however, for the random matrix ensembles studied here. For Wigner type matrices, in particular for the Gaussian band ensemble \eqref{eq:GBE}, we will obtain the estimates (\ref{eq:Weg1}, \ref{eq:Weg2}) with $\sigma = \half$ in \S\ref{sec:ensembles},.  
\item The parameters $\sigma$ and $a$ are not independent.  If we rescale via $V \mapsto W^{\gamma} V$ this results in a shift $\sigma \mapsto \sigma -\gamma$ and $a \mapsto a + \gamma$. Nonetheless it is convenient to keep both parameters since the \emph{natural} scaling of $V$ is to choose the eigenvalue spacing to be of order $1/W$.  This typically leads to $a=0$, but if the entries of $V$ have heavy tails then one may have $a > 0$.
\end{enumerate}
\end{rems}

We require very little from the distribution of $T_{j}$, denoted $\field{Q}_{W}$, essentially just a uniform (in $W$) bound on the tails:
\begin{ass}\label{ass:2} Let $( \field{Q}_{W})_{W=2}^{\infty}$ be a family of probability measures, with $\field{Q}_{W}$ supported on $\cu{T}_{W}$.  Suppose that there are $q_{0},\tau> 0$ and $b \ge 0$ such that
\begin{equation}\label{eq:normboundb} \field{Q}_{W} \set { T \ : \ \norm{T} \le \tau W^{b} } \ge q_{0}.
\end{equation}
\end{ass}
\begin{rem} $\field{Q}_{W}$ could be supported on a single point, in which case $T_{j}$ would be a constant sequence. For instance, we could take $T_{j}=I$.
\end{rem}

Lemma W for $X_{W;nW}$ follows easily from part (2) of assumption 1.
\begin{lem}[Lemma W for $X_{W;nW}$] Let $V_{j}$, $T_{j}$, $j=1, \ldots$, be mutually independent sequences of independent random $W \times W$ matrices.  Suppose each $V_{j}$ has distribution $\field{P}_{W}$ and each $T_{j}$ has distribution $\field{Q}_{W}$. Then, for each $\lambda \in \R$, 
$$ \Pr \left [ \lambda \text{ is an eigenvalue of $X_{W;nW}$} \right ] =0$$
and 
\begin{equation}\label{eq:averagedWeg1}
\Pr \left ( \norm{ P_{i} (X_{W;nW} - \lambda I)^{-1} P_{j}} > t W^{1+ \sigma} \ | \ \{T_{k} \}_{k=1}^{n-1} \text{ and } \{V_{k}\}_{k \neq i, j} \right ) \ \le \ 2 \kappa \frac{1}{t}
\end{equation}
for any $1\le i,j \le n$.
\end{lem}
\begin{proof} Let us first consider the case $i=j$.  The Schur complement formula shows that
$$ P_{i} (X_{W;nW} - \lambda I)^{-1} P_{i} = \left (V_{i} - \lambda I + K \right )^{-1}$$
with
$$ K = T_{i-1}^{\dagger} (X_{-} - \lambda I)^{-1} T_{i-1} + T_{i}(X_{+}-\lambda I )^{-1} T_{i}^{\dagger}$$ with $X_{-}$ and $X_{+}$ the restrictions of $X$ to the blocks above and below $i$.  By Lemma \ref{lem:algebra} $K \in \cu{A}_{W}^{H}$. (Note that it is self adjoint.)  It  follows from \eqref{eq:Weg1} that $\lambda$ is an eigenvalue of $X_{W;nW}$ with probability $0$ and that \eqref{eq:averagedWeg1} holds for $i=j$.

The argument for $i \neq j$ is similar.   In this case, we first estimate
$$ \norm{P_{i} (X_{W;nW} - \lambda I)^{-1} P_{j}} \le \norm{(P_{i}+ P_{j})  (X_{W;nW} - \lambda I)^{-1} (P_{i}+P_{j})}.$$
As above, we have
$$(P_{i}+ P_{j})  (X_{W;nW} - \lambda I)^{-1} (P_{i}+P_{j}) = \left [ \begin{pmatrix} V_{i} - \lambda I & 0 \\
0 & V_{j} - \lambda I \end{pmatrix} + \begin{pmatrix}A & C \\
C^{\dagger} & B \end{pmatrix} \right ],$$
where $A,B$ and $C$ are formed from blocks of the resolvents of restrictions of $X_{W;nW}$.  One may verify that $A,B \in \cu{A}_{W}^{H}$ and $C \in \cu{AT}_{W}.$ Thus the result follows from \eqref{eq:Weg2}.
\end{proof}

It follows that
\begin{equation}\label{eq:smomentband}
\Ev{ \norm{P_{i} (X_{W;nW} - \lambda I)^{-1} P_{j}}^{s}} \le \frac{2^{s}\kappa^{s}}{1-s} W^{(1+ \sigma)s},
\end{equation}
and so
\begin{equation}\label{eq:smomentbandmatelement}
\Ev{ \abs{ \ip{\vec{v}, P_{i} (X_{W;nW} - \lambda I)^{-1} P_{j} \vec{w}}}^{s}} \le \frac{2^{s}\kappa^{s}}{1-s} W^{(1+ \sigma)s},
\end{equation}
for any two vectors $\vec{v}, \vec{w}$. (See \eqref{eq:bathtub}.)

Lemma F in this context is as follows:
\begin{lem}[Lemma F for $X_{W;nW}$] \label{lem:FBand}Let $V_{j}$, $T_{j}$, $j=1, \ldots$, be mutually independent sequences of independent random $W \times W$ matrices.  Suppose each $V_{j}$ has distribution $\field{P}_{W}$ and each $T_{j}$ has distribution $\field{Q}_{W}$.   Let $D_{W}$ denote the real dimension of $\cu{A}_{W}^{H}$.  Fix a positive number $\nu$  large enough that $\sup_{W} D_{W} W^{-\nu} < \infty$ and suppose also that   $\nu \ge \zeta + \max(a, 1+\sigma + 2b) $, with $\sigma,  a, \zeta$ as in assumption 2 and $b$ as in assumption 3.  Then for each $0 < r < s <1$ and $0 < \Lambda < \infty$ there is $C_{r,s} > 0$ such that if $|i-j| \ge 3$ then
\begin{multline}
\Ev{ \left [ \Phi \left ( P_{i} (X_{W;nW} - \lambda I^{-1} P_{j}\right  ) \right ] ^{r}} \\ \le \exp\left(-C_{r,s} W^{-2\nu} |i-j|\right ) \Ev{ \left [ \Phi \left ( P_{i} (X_{W;nW} - \lambda I)^{-1} P_{j} \right  ) \right ]^{s}}^{r/s}
\end{multline}
for any $\lambda \in [-\Lambda, \Lambda]$ and any non-negative, positive-homogeneous function $\Phi: \cu{A T}_{W} \ra \R$ \tem \ i.e., $\Phi(Y) \ge 0$ and $\Phi(\alpha Y) = \alpha \Phi(Y)$ for $\alpha \ge 0$.
\end{lem}
\begin{rems} ~

\begin{enumerate} \item 
Below we will apply the result with $\Phi(Y)$ a semi-norm such as the the absolute value of a matrix element $\abs{\ip{\vec{v}, Y \vec{w}}}$ or the norm $\norm{Y}$.  However the proof does not make use of the triangle inequality, so the result also applies, for example, to $\Phi(Y)=$ spectral radius $(Y)$ or $\Phi(Y) = $ smallest singular value of $\Phi$.
\item Under rescaling of the matrix elements $X_{W;nW} \mapsto W^{\gamma} X_{W;nW}$ the localization length $1/C_{r,s} W^{-2\nu}$ should not change.  That this is indeed so follows since $\zeta \mapsto \zeta - \gamma$, $a \mapsto a + \gamma$, $\sigma \mapsto - \gamma$ and $b \mapsto b + \gamma$, so the combination $\zeta + \max(a, 1+\sigma + 2b)$ is invariant under rescaling.
\end{enumerate}
\end{rems}

Combining Lemma \ref{lem:FBand} and \eqref{eq:smomentbandmatelement} we have
\begin{thm}\label{thm:band} Let $V_{j}$, $T_{j}$, $j=1, \ldots$, be mutually independent sequences of independent random $W \times W$ matrices.  Suppose each $V_{j}$ has distribution $\field{P}_{W}$ and each $T_{j}$ has distribution $\field{Q}_{W}$. Let $D_{W}$ denote the real dimension of $\cu{A}_{W}^{H}$.  Fix a positive number $\nu$  large enough that $\sup_{W} D_{W} W^{-\nu} < \infty$ and suppose also that   $\nu \ge \zeta + \max(a, 1+\sigma + 2b) $, with $\sigma,  a, \zeta$ as in assumption 2 and $b$ as in assumption 3.  For $0 <t <1$ let 
\begin{equation}
M(W,t) = \max_{1 \le x,y \le nW}  \Ev{ |\left \langle \vec{e}_{x}, (X_{W;N} -\lambda)^{-1} \vec{e}_{y} \right \rangle|^{t} },
\end{equation}
where $\vec{e}_{x}$ and $\vec{e}_{y}$ denote elementary basis vectors.
Then 
\begin{equation}\label{eq:Mbound}
M(W,t) \le \frac{2^{t}\kappa^{t}}{1-t} W^{(1+ \sigma)t}
\end{equation}
and given $0 < s<t$ there are constants $C, \mu$ such that for any $1 \le x, y \le nW$
\begin{equation}\label{eq:thm:band}
\Ev{ |\left \langle \vec{e}_{x}, (X_{W;N} -\lambda)^{-1} \vec{e}_{y} \right \rangle|^{s} } \ \le \ C M(W,t)^{s/t} \e^{- \mu W^{-2\nu-1} |x-y|} .
\end{equation}
\end{thm}
\begin{proof} This amounts to special cases of \eqref{eq:smomentbandmatelement} and Lemma \ref{lem:FBand}.  The exponent $2\nu +1$ appears in \eqref{eq:thm:band} because the difference $|i-j|$ of the blocks to which $x$ and $y$ belong is estimated by $|x-y| /W$.  The constant $C$ compensates for the exponential factor $\e^{-\mu W^{-\nu-1} |x-y|}$ when $|x-y|$ is smaller than $3W$, in which case the estiamte of Lemma \ref{lem:FBand} does not hold.
\end{proof}
\begin{rem}Putting \eqref{eq:thm:band} and \eqref{eq:Mbound} together we have
\begin{equation}\label{eq:decayagain}
\Ev{ |\left \langle \vec{e}_{x}, (X_{W;N} -\lambda)^{-1} \vec{e}_{y} \right \rangle|^{s} } \le \const W^{(1 + \sigma) s} \e^{- \mu W^{-\nu-1} |x-y|} .
\end{equation}
If the diagonal blocks $V_{j}$ are Wigner matrices, as in assumption 4 in \S\ref{sec:ensembles} below, one may obtain the estimate
\begin{equation}
M(W,t) \le \frac{2^{t}\kappa^{t}}{1-t} W^{\frac{1}{2}},
\end{equation}
resulting in a very slight improvement on the estimate on the r.h.s.\ of \eqref{eq:decayagain},
\begin{equation}\label{eq:decayyetagain}
\Ev{ |\left \langle \vec{e}_{x}, (X_{W;N} -\lambda)^{-1} \vec{e}_{y} \right \rangle|^{s} } \le \const W^{\frac{s}{2} } \e^{- \mu W^{-\nu-1} |x-y|} .
\end{equation}
This improvement is not very significant, as the main point here is the exponential factor, which dominates any power of $W$ as long as $|x-y| >>W^{2\nu + 1}$.
\end{rem}

\section{Fluctuations}\label{sec:proof}
We now prove Lemma \ref{lem:FBand}.  Following the proof of Lemma \ref{lem:fluctW=2},  let us fix $\lambda$ and set
$$ G_{n}(i,j) = P_{i}(X_{W;nW} - \lambda I)^{-1} P_{j}.$$
Since $G_{n}(i,j)^{\dagger} = G_{n}(j,i)$, in estimating $\norm{G_{n}(i,j)}$ we may assume without loss that $i \le j$.  We have, by the resolvent identity,
\begin{equation}\label{eq:wg}
G_{n}(i,j)  \  = \ - G_{j-1}(i, j-1) T_{j-1} G_{n}(j,j).
\end{equation}
Iteration gives
\begin{equation}
G_{N}(i,j) \ = \ \left ( - 1\right )^{j-i} G_{i}(i,i) T_{i}G_{i+1}(i+1,i+1) T_{i+1} \cdots G_{j-1}(j-1,j-1) T_{j-1} G_{n}(j,j) . 
\end{equation}
Let us define $W \times W$ random matrices
\begin{equation}
\Gamma_{k}=  G_{k}(k,k)^{-1} ,
\end{equation}
related by a recursion relation 
\begin{equation}\label{eq:recurrenceW}
\Gamma_{k} \ = \ V_{k} -\lambda I -  T_{k-1}^{\dagger} \Gamma_{k-1}^{-1} T_{k-1} .
\end{equation}
As in the $W=2$ case, these identities may be established using the Schur-complement formula \tem \ compare with  \eqref{eq:w=2g} and \eqref{eq:recurrence}. Similarly,
\begin{equation}
G_{n}(j,j)^{-1}= V_{j} - \lambda I -T_{j-1}^{\dagger} \Gamma_{j-1}^{-1}T_{j-1} - T_{j}^{\dagger} \widehat G_{j+1} T_{j} 
= \Gamma_{j} - T_{j} \widehat G_{j+1} T_{j} ^{\dagger}
\end{equation}
where $\widehat G_{j+1} = P_{j+1} (\widehat X_{W;nW}- \lambda)^{-1} P_{j+1}$ with $\widehat X_{W;nW}$ the matrix obtained from $X_{W;nW}$ by setting $T_{j}=0$.  Thus $\widehat G_{j+1}$ is a function of the matrix variables $(V_{k})_{k=j+1}^{N}$ and $(T_{k})_{k=j+1}^{N}$.

We now make the change of variables $V_{k} \mapsto \Gamma_{k}$ in our probability space. By Lem.\ \ref{lem:algebra} and Prop.\ \ref{prop:cstar}, $\Gamma_{k} \in \cu{A}_{W}^{H}$.  As in the tri-diagonal case, the Jacobian determinant is $1$, so 
\begin{multline}
\text{Joint distribution of $(\Gamma_{k})_{k=1}^{n}$ given $(T_{k})_{k=1}^{n-1}$} \\ = \ 
\rho(\Gamma_{1} + \lambda I)
\prod_{k=2}^n \rho(\Gamma_k + \lambda I + T_{k-1}^{\dagger} \Gamma_{k-1} T_{k-1} )  \, \prod_{k=1}^n \di \Gamma_k,
\end{multline} 
where $\di \Gamma_{k}$ denotes Lebesgue measure on $\cu{A}_{W}^{H}$.
In terms of the matrices $\Gamma_{k}$, we have
\begin{equation}\label{eq:Grep}
G_{n}(i,j) =  (-1)^{|i-j|} \Gamma_{i}^{-1} T_{i} \Gamma_{i+1}^{-1} T_{i+1} \cdots \Gamma_{j-1}^{-1} T_{j} \cdot  (\Gamma_{j} - T_{j} \widehat G_{j+1} T_{j} ^{\dagger} )^{-1}, 
\end{equation}
where $\widehat G_{j+1}$ is a function of $(\Gamma_{k})_{k=j}^{n}$ and $(T_{k})_{k=j}^{n}$ (since $V_{k} = \Gamma_{k} +  \lambda I - T_{k-1}^{\dagger} \Gamma_{k-1}^{-1} T_{k-1}$).

The matrix product in \eqref{eq:Grep} is non-commutative, so it is not clear if the heuristic analysis that the ``log of $G$ is a sum of terms with only local correlations'' is valid. Nonetheless, we may use the trick employed above of coupling the system to a family of independent identically distributed scalar variables $\alpha_{2}, \alpha_{5}, \ldots$, each with absolutely continuous distribution \begin{equation}
H(\alpha_{k}) \di \alpha_{k} \ = \ \frac{1}{2\lf} I[|\alpha_{k}| \le \lf]  \di \alpha_{k} , 
\end{equation}with $\lf > 0$ to be chosen below.  We define 
\begin{equation}
F_k =  \e^{\alpha_k} \Gamma_k ,
\end{equation}
where we take $\alpha_{k}=0$ for $k \not \equiv 2 \mod 3$. The Jacobian of the transformation $(\Gamma_{k},\alpha_{k}) \mapsto (F_{k}, \alpha_{k})$ is $\prod_{k=2,5,8,\ldots}^{N} \e^{-D_{W}\alpha_k}$, where $D_{W} = \dim \cu{A}_{W}^{H}$ is the  dimension of $\cu{A}_{W}^{H}$.
  Thus
\begin{multline}
\text{joint distribution of $(F_k)_{k=1}^{n}$ and $(\alpha_k)_{k=1}^{N}$, given $(T_{k})_{k=1}^{n-1}$} \ = \\
\prod_{k\equiv 2 \mod 3}^{n} \rho(F_{k-1} + \lambda I + T_{k-2}^{\dagger}F_{k-2}^{-1} T_{k-2}) \rho(\e^{-\alpha_{k}}F_{k} + \lambda I + T_{k-1}^{\dagger}F_{k-1}^{-1} T_{k-1})\\ \times \rho(F_{k+1} +\lambda I + \e^{\alpha_{k}} T_{k}^{\dagger}F_{k}^{-1} T_{k}) \, H(\alpha_{k}) \e^{- D_{W}\alpha_{k}}   \di F_{k-1} \di F_{k} \di F_{k+1} \di \alpha_{k},\end{multline}
with the convention that $T_{0}=0$.  

As in the tri-diagonal case, the variables $\alpha_k$ remain independent after conditioning on $(F_{k})_{k=1}^{N}$.  Also, the
\begin{multline}
\text{distribution of $\alpha_{k}$ given $(T_{\ell})_{\ell=1}^{n-1}$ and $(F_{\ell})_{\ell=1}^{n}$} \ = \\
 \frac{ \rho(\e^{-\alpha_k} F_k +  \lambda I + T_{k-1}^{\dagger}F_{k-1}^{-1} T_{k-1} ) \rho(F_{k+1} +\lambda I +  \e^{\alpha_{k}} T_{k}^{\dagger}F_{k}^{-1} T_{k}) H(\alpha_k) \e^{-D_{W}\alpha_k} }{Z_k} \di \alpha_{k}
\end{multline}
with 
\begin{equation}
Z_k \ = \ \frac{1}{2\lf} \int_{-\lf}^{\lf} \rho(\e^{-\alpha} F_k + \lambda I + T_{k-1}^{\dagger}F_{k-1}^{-1} T_{k-1} ) \rho(F_{k+1} +\lambda I +   \e^{\alpha}T_{k}^{\dagger}F_{k}^{-1} T_{k}) \e^{-D_{W}\alpha} \di \alpha.
\end{equation}

Now fix a non-negative positive homogeneous $\Phi$ as in the statement of the Lemma.   Replacing $\Gamma_{j}$ in \eqref{eq:Grep} by $\e^{-\alpha_{j}} F_{j}$, we find that
\begin{equation}
\Phi(G_{n}(i,j)) = \prod_{\substack{k \equiv 2 \mod 3  \\ i \le k \le j-1}} \e^{\alpha_{k}}  \Phi \left ( (-1)^{|i-j|}  \left (\prod_{k=i}^{j-1} F_{k}^{-1} T_{k} \right)  \widehat H_{j+1} \right ) ,
\end{equation}
where
\begin{equation}
\widehat H_{j+1} = 
\frac{1}{\Gamma_{j} -  T_{j} \widehat G_{j+1} T_{j} ^{\dagger}}
\end{equation}
is a function of $(T_{k},F_{k},\alpha_{k})_{k=j}^{N}$. 
Since $(\alpha_{k})$ are conditionally independent, it follows that
\begin{multline}\Ev{\left . \left [ \Phi(G_{n}(i,j)) \right] ^{r} \right | (T_{\ell},F_{\ell})_{\ell=1}^{N}, \ (\alpha_{\ell})_{k=j}^{N}} \\ = \ \left [ \Phi \left ( (-1)^{|i-j|}   \left (\prod_{k=i}^{j-1} F_{k}^{-1} T_{k} \right)  \widehat H_{j+1} \right ) \right ]^{r}   \prod_{\substack{k \equiv 2 \mod 3  \\ i \le k \le j-1}} \Ev{\left . \e^{r\alpha_{k}} \right |(T_{k},F_{k})_{k=1}^{N}} .
\end{multline}
By propostion \ref{prop:basic} and the H\"older inequality, we conclude that (compare with \eqref{eq:keyresult}):
\begin{equation}
\Ev{\left [ \Phi(G_{n}(i,j)) \right] ^{r}} \ \le \ \Ev{\e^{- \frac{s}{s-r} \sum_{k=i}^{j-1} h_{k}(r,s)}}^{\frac{s-r}{s}}  \Ev{\left [ \Phi(G_{n}(i,j)) \right] ^{s}}^{r/s}, \label{eq:matkeyresult}
\end{equation}
where for $k \equiv 2 \mod 3$
\begin{equation}\label{eq:hintegralII}
h_{k}(r,s) = \frac{1}{s} \int_{0}^{s} \min(r,q) (s- \max(r,q)) \Var_{q}(\alpha_{k} |  (T_{\ell},F_{\ell})_{\ell=1}^{N}) \di q ,
\end{equation}
with $\Var_{q}$ as in \eqref{eq:conditionalvariance}, and we have set
$h_{k}(r,s)=0$ for $k \not \equiv 2 \mod 3$.

Let us express $ \Var_{q}(\alpha_{k} |  (T_{\ell},F_{\ell})_{\ell=1}^{N})$  in terms of $(T_{\ell}, \Gamma_{\ell}, \alpha_{\ell})_{\ell=1}^{N}$:
\begin{equation}
 \Var_{q}(\alpha_{k} |  (T_{\ell},F_{\ell})_{\ell=1}^{N}) \ = \ \inf_{m \in \R} \frac{\int_{-\lf}^{\lf} (\alpha- m)^{2} \e^{(q-D_{W}) \alpha}\nu_{k}(\alpha) \di \alpha}{\int_{-\lf}^{\lf} \e^{(q-D_{W})\alpha} \nu_{k}(\alpha) \di \alpha}
\end{equation}
with 
\begin{equation}\label{eq:nuk}
\nu_{k} (\alpha) \ = \ \rho(\e^{\alpha_{k} - \alpha} \Gamma_{k} +\lambda I +T_{k-1}^{\dagger}\Gamma_{k-1}^{-1} T_{k-1}) \rho(\Gamma_{k+1} +  \lambda I  +\e^{\alpha - \alpha_{k}} T_{k}^{\dagger }\Gamma_{k}^{-1} T_{k}) .
\end{equation}
Thus (compare with  \eqref{eq:factoredlowerbound}), \begin{multline}\label{eq:matfactoredlowerbound}
 \Var_{q}(\alpha_{k} | (T_{\ell},F_{\ell})_{\ell=1}^{N}) \  \ge \ \frac{1}{3} \lf^{2} \e^{-2q\lf} \e^{-2D_{W} \lf}  \inf_{-2\lf < \alpha, \beta < 2\lf} \frac{ \rho(V_{k} +(\e^{-\alpha} -1) \Gamma_{k}) }{ \rho(V_{k} + (\e^{-\beta}-1)  \Gamma_{k})} \\ \inf_{-2\lf < \alpha, \beta < 2\lf}\frac{\rho(V_{k+1} + (\e^{\alpha}-1) T_{k}^{\dagger }\Gamma_{k}^{-1} T_{k}) }{ \rho(V_{k+1} + (\e^{\beta}-1) T_{k}^{\dagger }\Gamma_{k}^{-1} T_{k})} .
\end{multline}
With  \eqref{eq:matkeyresult} this implies 
\begin{equation}\label{eq:mat keyresultredux}
\Ev{\left [ \Phi(G_{n}(i,j)) \right] ^{r}} \ \le \ \Ev{\e^{-  \frac{r s}{6} \lf^2 \e^{-2s \lf} \e^{-2D_{W} \lf} \sum_{k=i}^{j-1} U_k(\lf)}}^{\frac{s-r}{s}}  \Ev{\left [ \Phi(G_{n}(i,j)) \right] ^{s}}^{r/s},
\end{equation}
where
\begin{equation}
U_k(\lf) = \inf_{-2\lf < \alpha, \beta < 2\lf} \frac{ \rho(V_{k} +(\e^{-\alpha} -1) \Gamma_{k}) }{ \rho(V_{k} + (\e^{-\beta}-1)  \Gamma_{k})} \inf_{-2\lf < \alpha, \beta < 2\lf}\frac{\rho(V_{k+1} + (\e^{\alpha}-1) T_{k}^{\dagger }\Gamma_{k}^{-1} T_{k}) }{ \rho(V_{k+1} + (\e^{\beta}-1) T_{k}^{\dagger }\Gamma_{k}^{-1} T_{k})} .
\end{equation}

By fluctuation regularity of $\bb{P}_{W}$, we have $U_{k}(\lf) \ge \delta^{2} I[A_{k}] $ where $I[A_{k}]$ is the indicator function of the event:
\begin{equation}
A_{k} = \set{V_{k}, V_{k+1} \in \Omega_{W} \ , \quad \norm{\Gamma_{k}} \le \frac{\epsilon }{ \e^{2\lf} -1} W^{-\zeta}, \ \ \text{and} \ \ \norm{T_{k}^{\dagger} \Gamma_{k}^{-1}T_{k}} \le \frac{\epsilon}{ \e^{2\lf} -1} W^{-\zeta}},
\end{equation}
with $\delta,\epsilon >0$, $\zeta \ge0$ and $\Omega_{W}$ as in assumption (3).   In turn, since $\Gamma_{k} = V_{k} + \lambda +T_{k-1}^{\dagger}\Gamma_{k-1}^{-1} T_{k-1} $, we see that
\begin{multline}\label{eq:matAksubset}
A_{k} \supset \set{V_{k+1} \in \Omega_{W}} \cap\set{V_{k} \in \Omega_{W} } \cap  \set{ \norm{T_{k-1}}, \norm{T_{k}} \le \tau W^{b}} \\ \cap \set{\norm{\Gamma_{k-1}^{-1}}  \le \frac{1}{\tau^{2}}  \left ( 
 \frac{\epsilon}{ \e^{2\lf} -1} W^{-\zeta} - L W^{a} - |\lambda| \right )W^{-2b} } \\ \cap  \set{\norm{\Gamma_{k}^{-1}}  \le \frac{1}{\tau^{2}}  \frac{\epsilon}{ \e^{2\lf} -1}W^{-2b-\zeta}} ,\end{multline}
with $L, a \ge 0$ as in assumption 2 and $\tau,b\ge 0$ as in assumption 3.  This allows us to estimate the probability of $A_{k}$ from below by successively integrating over $V_{k+1}$, $V_{k}$, $V_{k-1}$, $T_{k}$, and $T_{k-1}$ in that order.   To begin, by assumption 2, 
\begin{equation}
\Pr(V_{k+1} \in \Omega_{W} | (V_{l})_{l \neq k} , \ (T_{l}) ) = \bb{P}_{W}(\Omega_{W}) \ge p_{0} >0 .\end{equation}
Since $\Gamma_{k} = V_{k} + \lambda I + T_{k-1}^{\dagger}\Gamma_{k-1}^{-1} T_{k-1}$, we see from the Wegner estimate \eqref{eq:Weg1} that
\begin{multline}
\Pr \left ( \left . V_{k} \in \Omega_{W}\ , \ \norm{\Gamma_{k}^{-1}}  \le \frac{1}{\tau^{2}}  \frac{\epsilon}{ \e^{2\lf} -1} W^{-2b -\zeta} \right | (V_{l})_{l \neq k,k+1} , \ (T_{l}) \right )  \\ \ge 
p_{0} - \kappa  \tau^{2}  \frac{\e^{2\lf}- 1}{\epsilon} W^{1+ \sigma + \zeta + 2b}.
\end{multline}
Similarly
\begin{multline}
\Pr \left ( \left. \norm{\Gamma_{k-1}^{-1}} \le \frac{1}{\tau^{2}} \left ( 
 \frac{\epsilon}{ \e^{2\lf} -1}W^{-\zeta} - L W^{a} - | \lambda|\right ) W^{-2b} \right | (V_{l})_{l \neq k-1,k,k+1}, \ (T_{l}) \right )\\ \ge 1 -  \kappa \tau^{2} \frac{1}{\frac{\epsilon}{ \e^{2\lf} -1}W^{-\zeta} - L W^{a} - |\lambda|} W^{1+\sigma+2b}.
\end{multline}
Combining these estimates and using assumption 3 to integrate over $T_{k}$ and $T_{k-1}$, we find
\begin{multline}\label{eq:Akest}
\Pr(\left . A_{k} \right | (V_{l})_{l \neq k-1,k,k+1} , \ (T_{l})_{l \ne k, k-1}) \\ \ge \ q_{0}^{2} p_{0}\left ( p_{0} - \kappa  \tau^{2}  \frac{\e^{2\lf}- 1}{\epsilon} W^{1 + \sigma + 2b+\zeta} \right) \left ( 1 -   \kappa \tau^{2} \frac{1}{\frac{\epsilon}{ \e^{2\lf} -1}W^{-\zeta} - L W^{a} - |\lambda|} W^{1+\sigma+2b} \right ).
\end{multline}

Taking $\lf = c W^{-\nu}$ with $\nu \ge \max( a , 2b+\sigma+1) + \zeta$, we may choose $c$ sufficiently small to make the r.h.s. larger than $ \half q_{0}^{2} p_{0}^{2}$, say.  Since $U_{k} (\eta) \ge \delta^{2} I[A_{k}]$ we find, integrating successively over $V_{k}, T_{k}$ from $k=i, \ldots, j-1$ (see Lemma \ref{lem:prob}), that
\begin{multline}
\Ev{\e^{-  \frac{r s}{6} \lf^2 \e^{-s\lf} \e^{-D_{W}\lf} \sum_{k=i}^{j-1} U_k(\lf)}}^{\frac{s-r}{s}}
\\ \le \ \exp\left ( - \frac{s-r}{2s} q_{0}^{2}p_{0}^{2}  \left (1 - \e^{- c^{2} \delta^{2}\frac{r s}{6} W^{-2 \nu} \e^{-c sW^{-\nu}} \e^{-c D_{W} W^{-\nu }}} \right ) \left \lfloor \frac{|i-j|}{3} \right \rfloor \right ).
\end{multline}
Increasing $\nu$, if necessary, so that $\sup_{W}D_{W} W^{-\nu} < \infty$ completes the proof. \qed

\section{Ensembles}\label{sec:ensembles}
In this section, we consider several examples of band matrix ensembles satisfying assumptions 1, 2, and 3 of section \ref{sec:axioms}.   Assumption 1 is simply the choice of an algebra $\cu{A}_{W}$ to support the distribution of the diagonal blocks, and the corresponding set $\cu{T}_{W}$ for the off-diagonal blocks.  In this regard, we will consider two cases:
\begin{enumerate}
\item[($\R$)] $\cu{A}_{W}=$ $W \times W$ matrices with real entries,

or

\item[($\C$)] $\cu{A}_{W}=$ $W \times W$ matrices with complex entries.
\end{enumerate}
In each case the dimension of the algebra $D_{W}$ is comparable to $W^{2}$ and $\cu{T}_{W} = \cu{A}_{W}$. 

\subsection{Wigner-matrix blocks and the Wegner estimate} We shall suppose that the diagonal blocks $V_{j}$ of $X_{W;N}$ are \emph{Wigner matrices}:
\begin{ass}The distribution of the diagonal blocks, $\di \bb{P}_{W}(V)$, written in terms of the matrix elements 
\begin{equation}\label{eq:Wigner}
V \ = \ \frac{1}{\sqrt{W}} \begin{pmatrix} d_{1} & a_{1,2} & \cdots  & \cdots &a_{1,W} \\
a_{1,2}^{*} & d_{2} & & &  \vdots \\
\vdots && \ddots & & \vdots \\
\vdots & & & \ddots & a_{W-1,W}\\
a_{1,W}^{*}&\cdots &\cdots & a_{W-1,W}^{*} & d_{W}
\end{pmatrix} ,
\end{equation}
has the form
\begin{equation}
\di \bb{P}_{W}(V) \ = \  \prod_{j=1}^{W} h(d_{j}) \di d_{j} \ \prod_{1\le i < j\le W} g( a_{i,j})  \di a_{i,j},
\end{equation}
where $\di d_{j}$ is Lebesgue measure on the real line, $h \in L^{\infty}(\R) \cap L^{1}(\R)$ is non-negative with $\int h = 1$, and either
\begin{enumerate}
\item[($\R$)] $\di a_{i,j}$ is Lebesgue measure on $\R$ and $g\in L^{\infty}(\R) \cap L^{1}(\R)$ is non-negative with $\int_{\R} g = 1$,

or

\item[($\C$)] $\di a_{i,j}$ is Lebesgue measure on $\C$ and $g \in L^{\infty}(\C) \cap L^{1}(\C)$ is non-negative with $\int_{\C} g= 1$.
\end{enumerate}
Furthermore, we require
\begin{equation}
\int_{\R} \lambda^{2} h(\lambda) \di \lambda  \ < \ \infty ,
\end{equation}
\begin{equation}
\int  |a|^{4} g(a) \di a \ < \ \infty \ , \quad \text{and} \quad \int a g(a) \di a = 0 .
\end{equation}
\end{ass}

Clearly the measure $\di \bb{P}_{W}$ is absolutely continuous with respect to Lebesgue measure on $\cu{A}_{W}^{H}$ \tem \ this is part 1 of assumption 2. Regarding the Wegner estimates \tem \ part 2 of assumption 2 \tem \ we then have the following
\begin{thm}[Wegner estimate]\label{thm:Wegner} Under assumption 4, the Wegner estimates \eqref{eq:Weg1} and \eqref{eq:Weg2} hold with $\sigma = \frac{1}{2}$ and $\kappa = 2 \pi \esssup_{\lambda}h(\lambda).$
\end{thm}
\begin{proof} This result, which is obtained by averaging over the diagonal variables $\{d_{j}\}$ only, is a standard estimate from the theory of random Schr\"odinger operators, first obtained by Wegner \cite{Wegner:1981p5335}.  For completeness,  we sketch the proof.  

Note that $\norm{(V-A)^{-1}} > t$ if and only if $V-A$ has an eigenvalue in the interval $(-\frac{1}{t}, \frac{1}{t})$.  It follows that
\begin{multline} \Pr \set{\norm{(V-A)^{-1}} > t} \le \frac{2}{t^{2}}  \Ev{ \tr \left [ (V-A)^{2} + \frac{1}{t^{2}}\right ]^{-1}} \\ 
 = \ \frac{2}{t} \Ev{ \Im \tr \left [ V - A - \im \frac{1}{t} I\right ]^{-1}}
 = \ \frac{2}{t} \sum_{i=1}^{W} \Ev{ \Im \ip{\vec{e}_{i},  \left [ V - A - \im \frac{1}{t} I\right ]^{-1} \vec{e}_{i}}}.
\end{multline}
By the Schur complement formula,
\begin{equation}
\ip{\vec{e}_{i},  \left [ V - A - \im \frac{1}{t} I\right ]^{-1} \vec{e}_{i}}
= \frac{1}{\frac{1}{\sqrt{W}} d_{i} - \frac{1}{t} \im - \gamma} ,
\end{equation}
where $\gamma$ is a function of all matrix elements of $V$ \emph{except} $d_{i}$.  Thus $\gamma$ is a random variable independent of $d_{i}$, so
\begin{equation} 
 \Ev{ \Im \ip{\vec{e}_{i},  \left [ V - A - \im \frac{1}{t} I\right ]^{-1} \vec{e}_{i}}}
 \ = \ \Ev{ \frac{\frac{1}{t} + \Im \gamma}{(\frac{1}{\sqrt{W}}d_{i} - \Re \gamma)^{2} + (\frac{1}{t} + \Im \gamma)^{2}}} \ \le \ \norm{h}_{\infty} \pi \sqrt{W},
\end{equation}
where the inequality follows from replacing the average $\int \bullet h(  d_{i}) \di d_{i}$ by the upper bound $\norm{h}_{\infty}\int_{\R} \bullet \di d_{i}.$  Summing over $i$ gives the result.

The proof of \eqref{eq:Weg2} is analogous.  However in that case the trace is over a $2W$ dimensional space, resulting in the additional factor of $2$ on the r.h.s.\ of that equation.
\end{proof}

The scaling factor $\sqrt{W}$ that appears in \eqref{eq:Wigner}  is natural, as with this scaling  the matrix $V$ has a finite density of states in the large $W$ limit \cite{Molchanov:1992p5303}:
\begin{equation}
\lim_{W \ra \infty } \frac{1}{W} \int_{ \cu{A}_{W}^{H}}\tr f(V) \di \bb{P}_{W}(V) \ = \ \frac{1}{2\sigma^{2} \pi} \int_{-2\sigma}^{2\sigma}  f(\lambda) \sqrt{4 \sigma^{2} - \lambda^{2}} \di \lambda ,
\end{equation}
with $\sigma^{2} = \int |a|^{2} g(a)\di a.$   A key fact below is the following related result
\begin{thm}[Bai and Yin \cite{Bai:1988hy}]  
Let $V$ be a $W \times W$ random matrix of the form \eqref{eq:Wigner}, with $\{d_{i}\}$ and $\{a_{i,j}\}$ mutually independent sets of independent random variables.  If
$$ \Ev{|d_{i}|} < \infty , \quad \Ev{|a_{i,j}|^{4}} < \infty, \quad \Ev{a_{i,j}}=0, $$
and $\sigma^{2} = \Ev{|a_{i,j}|^{2}}$ then
\begin{equation}\label{eq:normtozero} \lim_{W \ra \infty} \Pr \left [ \norm{V} > 2 \sigma + \eta \right ] = 0.
\end{equation}
\end{thm}
\begin{rem} This follows from Theorem A of ref.\ \cite{Bai:1988hy}, which gives the convergence of $\lambda_{1}$, the largest eigenvalue of $V$, to $2 \sigma$ with probability one.  Symmetrizing the assumptions of Theorem A and applying the result also to show that $\lambda_{W}$, the smallest eigenvalue of $V$, converges to $-2 \sigma$, this result follows.  (The proof in \cite{Bai:1988hy} is written out in the real symmetric case, but carries over to the complex hermitian case with only very minor modifications.)
\end{rem}

\begin{cor}\label{cor:uniformnormbound} Under assumption 4, we may find $p_{0}, L > 0$  such that
\begin{equation}\label{eq:uniformnormbound}
\Pr \left [ \norm{V} \le L \right ] \ \le \ p_{0}.
\end{equation}
\end{cor}

We require very little of the off diagonal blocks $T_{j}$.  They need only satisfy the estimate \eqref{eq:normboundb} analogous to \eqref{eq:uniformnorma} and \eqref{eq:uniformnormbound}.  In particular, they could be deterministic, say $T_{j}=I$ for all $j$ or $T_{j}$ given by a Toeplitz matrix.   In this section we consider a few examples of random off-diagonal blocks modeled on the blocks for the Gaussian band ensemble \eqref{eq:GBE}.  In that case, the off-diagonal blocks $T_{j}$ are lower triangular matrices with Gaussian entries.  More generally we may suppose
\begin{ass}The distribution of the off-diagonal blocks, $\di \bb{Q}_{W}(T)$, written in terms of the matrix elements 
\begin{equation}\label{eq:Wigner}
T \ = \ \frac{1}{\sqrt{W}} \begin{pmatrix} 0 & 0 & \cdots  & \cdots &0  \\
t_{2,1} & \ddots & & &  \vdots \\
\vdots &\ddots & \ddots & & \vdots \\
\vdots & &\ddots & \ddots & 0 \\
t_{W,1}&\cdots &\cdots & t_{W,W-1} & 0\end{pmatrix} ,
\end{equation}
has the form
\begin{equation}
\di \bb{Q}_{W}(T) \ = \  \prod_{1\le j < i\le W} \di \mu( t_{i,j}) ,
\end{equation}
where either
\begin{enumerate}
\item[($\R$)] $ \mu(t_{i,j})$ is a probability measure on $\R$ 

or

\item[($\C$)] $ \mu(t_{i,j})$ is a probability measure on $\C$,
\end{enumerate}
and 
\begin{equation}
\int  |t|^{4} \di\mu(t)  \ < \ \infty \ , \quad \text{and} \quad \int t \di \mu(t) = 0 .
\end{equation}
\end{ass}
\begin{thm}\label{thm:uniformTnorm} Under assumption 5, we may find $q_{0}, \tau > 0$ such that assumption 3 holds with $b=0$, i.e.,
\begin{equation}\label{eq:uniformTnorm} 
\Pr \left [ \norm{T} \le \tau \right ] \le q_{0}.
\end{equation}
\end{thm}
\begin{proof}It follows from \cite[Theorem A]{Bai:1988hy} that, with $\sigma^{2}= \int |t|^{2} \di \mu(t)$, 
\begin{equation}
\lim_{W \ra 0} \Pr \left [ \norm{T + T^{\dagger}} > \sigma + \eta \right ] =0, \quad \lim_{W \ra 0} \Pr \left [ \norm{\im( T - T^{\dagger})} > \sigma + \eta \right ] = 0,
\end{equation}
for any $\eta > 0$.  Since 
\begin{equation}
T \ = \ \half ( T + T^{\dagger} ) + \frac{1}{2\im} \im (T - T^{\dagger}) ,
\end{equation}
it follows that 
\begin{equation}
\lim_{W \ra 0} \Pr \left [ \norm{ T} >  \sigma + \eta \right ] = 0.
\end{equation}
Thus \eqref{eq:uniformTnorm} holds.
\end{proof}

\subsection{Fluctuation regularity} 
A particular example of distributions satisfying assumption 4 are the Gaussian Orthogonal Ensemble (GOE), corresponding to case ($\R$), and the Gaussian Unitary Ensemble (GUE), corresponding to case ($\C$).  In these cases, the measure $\bb{P}$ is of the form
\begin{equation}   
\di \bb{P}(V) \ \propto \  \e^{- \beta W \tr V^{2}} \di V, \quad V \in \cu{A}_{W}^{H},
\end{equation}
with $\beta =1$ ($\R$) or $\beta =2$ ($\C$).
That is,
\begin{equation}
h(d) \ = \ \frac{1}{\sqrt{\pi}} \e^{-\beta d^{2}}, \quad g(a) \ = \ \frac{1}{(2 \beta\pi)^{\frac{\beta}{2}}} \e^{-2 \beta |a|^{2}} .
\end{equation}
\begin{thm} If $V$ is a GUE or GOE matrix of size $W$ then assumption 2 of section \ref{sec:axioms} holds with $\sigma = \frac{1}{2}$, $\zeta =2$ and $a=0$. \end{thm}
\begin{cor}Assumptions 1, 2, and 3 hold for the Gaussian band ensemble \eqref{eq:GBE}.
\end{cor}
\begin{proof}We have already derived the Wegner estimates (Thm.\ \ref{thm:Wegner}).  It remains only to show the fluctuation regularity. For the Gaussian ensembles, we have
\begin{equation}
\frac{\rho(V_{1})}{\rho(V_{2})} \ = \ \e^{-\beta W \tr (V_{1}^{2} - V_{2}^{2})}
= \e^{-\beta W \tr(V_{1}-V_{2}) (V_{1}+V_{2})} \ge \e^{-\beta W^{2} \norm{V_{1}-V_{2}} \norm{V_{1}+ V_{2}}} .
\end{equation}
If $\norm{V_{1}- V}, \norm{V_{2}- V} \le \epsilon W^{-2}$ we have
\begin{equation}
\frac{\rho(V_{1})}{\rho(V_{2})}  \ge \e^{-2 \beta \epsilon ( \norm{V} + \epsilon W^{-2}) }.
\end{equation}

Letting $p_{0}$ and  $L$ be as in Cor.\ \ref{cor:uniformnormbound}, we set $\Omega_{W} := \{ \norm{V} \le L\}$.  Then $\Pr(\Omega_{W}) \ge p_{0} > 0$ and if $V \in \Omega_{W}$, we have
\begin{equation}
\frac{\rho(V_{1})}{\rho(V_{2})}  \ge \e^{-2 \epsilon ( L + \epsilon ) } \ := \ \delta,
\end{equation}
whenever $\norm{V_{1}- V}, \norm{V_{2}- V} \le \epsilon W^{-2}$. \end{proof}

To obtain fluctuation regularity for general Wigner matrices \eqref{eq:Wfluctreg} we require additional assumptions on $h$ and $g$. For instance, we have the following
\begin{thm}\label{thm:unHcont} If $V$ satisfies assumption 4 with $\ln h$ and $\ln g$ uniformly  H\"older continuous with exponent $\alpha$, then assumption 2  of section \ref{sec:axioms} holds with $\sigma = \frac{1}{2}$, $\zeta =\frac{2}{\alpha} + \frac{1}{2}$ and $a=0$.\end{thm}
\begin{rem} For example $h(\lambda) = g(\lambda) = c_{\alpha} \e^{-|\lambda|^{\alpha}}$ with $0 < \alpha \le 1$ satisfies the hypotheses of the theorem.
\end{rem}
\begin{proof}
We have 
\begin{multline}
\frac{\rho(V_{1})}{\rho(V_{2})} \ = \ \exp \left ( \sum_{i} \ln h( d_{i;1}) - \ln h( d_{i;2}) \ + \ \sum_{i,j} \ln g(  a_{i,j;1}) - \ln g(  a_{i,j;2}) \right ) \\
\ge \exp\left [ - C \left ( \sum_{i} |d_{i;1} - d_{i;2}|^{\alpha} + 
\sum_{i,j} \ln |a_{i,j;2} - a_{i,j;2}|^{\alpha} \right ) \right ]
\end{multline}
If $\norm{V_{1}- V}, \norm{V_{2}- V} \le \epsilon W^{-\zeta}$, then $$|d_{i;1} - d_{i;2}| , \ |a_{i,j;2} - a_{i,j;2}| \ \le \ \sqrt{W} \norm{V_{1}- V} + \sqrt{W} \norm{V_{2}- V}
\ \le \ 2 \epsilon W^{\half - \zeta}.$$
It follows that
\begin{equation}
\frac{\rho(V_{1})}{\rho(V_{2})} \ \ge \ \exp\left [ -  C \epsilon W^{2 }W^{ \frac{\alpha}{2}-\zeta \alpha} \right ]   \ = \ \e^{- C \epsilon} \ =: \ \delta.
\end{equation}
This estimate holds for \emph{every} $V$, so in particular for all $V$ in $\Omega_{W} = \{ \norm{V} \le L\}$.
\end{proof}

Theorem \ref{thm:unHcont} cannot apply if $h$ or $g$ has compact support. Nonetheless compactly supported densities can be handled.  A general result of this type would somewhat involve to state, so let us simply note that assumption 2 holds if $h$ and $g$ are characteristic functions of open neighborhoods of the origin.
\begin{thm}\label{thm:charact}  Suppose that $V$ satisfies assumption 4 and that 
$$ h(d) = \frac{1}{2 D} I[ |d| < D], \quad g(a) = \frac{1}{c_{\beta} A^{\beta}} I[ |a| < A],$$
with $c_{1}= 2$ and $c_{2}= \pi$.  Then assumption 2  of section \ref{sec:axioms} holds with $\sigma = \frac{1}{2}$, $\zeta = \frac{5}{2}$ and $a=0$.
\end{thm}
\begin{proof}  Clearly the moment conditions of assumption 4 hold.  Thus by Cor.\ \ref{cor:uniformnormbound} we can find $p_{0}$ and $L$ so that \eqref{eq:uniformnormbound} holds.  

Now suppose $\norm{V_{1}- V}, \norm{V_{2}- V} \le \epsilon W^{-\frac{5}{2}}$.  Suppose also that the matrix elements of $V$ satisfy $ W^{-\half }|d_{i}| \le W^{-\half }D - \epsilon W^{-\frac{5}{2}}$, $ W^{-\half }|a_{i,j}| \le W^{-\half }A - \epsilon W^{-\frac{5}{2}}$  for all $i,j$.    Then 
\begin{equation}
\frac{\rho(V_{1})}{\rho(V_{2})} = 1.
\end{equation}
But 
\begin{multline}
\Pr( W^{-\half }|d_{i}| \le W^{-\half }D - \epsilon W^{-\frac{5}{2}}, \   W^{-\half }|a_{i,j}| \le W^{-\half }A - \epsilon W^{-\frac{5}{2}} )
\\ \ge \  1 - \sum_{i} \Pr ( |d_{i}| >  D - \epsilon W^{-2 } ) - 
\sum_{i,j} \Pr (  |a_{i,j}| > A - \epsilon W^{-2} )  \ge 1 - C \epsilon . 
\end{multline}
Now let 
\begin{equation}
\Omega_{W} \ = \ \set{ \norm{V} \le L \ , \ W^{-\half } |d_{i}| \le W^{-\half }D - \epsilon W^{-\frac{5}{2}} \ , \ \text{ and } \  W^{-\half }|a_{i,j}| \le W^{-\half }A - \epsilon W^{-\frac{5}{2}}} ,
\end{equation}
with $\epsilon$ sufficiently small that \begin{equation}
\Pr (\Omega_{W}) \ge p_{0} - C \epsilon > 0 . \qedhere \end{equation}
\end{proof}

\subsection{Summary}Putting the results of this section together with Thm.\ \ref{thm:band} we have:
\begin{thm}. 
Let $\cu{A}_{W}= \cu{T}_{W} =$ set of $W \times W$ matrices with real or complex entries  and suppose $\bb{P}$ and $\bb{Q}$ satisfy assumptions 4 and 5.    
\begin{enumerate}
\item If $\bb{P}$ is either the Gaussian orthogonal or Gaussian unitary ensemble, then
given $r >0$ and $s  \in (0,1)$ there are $A_s < \infty$ and $\alpha_s > 0$  such
that
\begin{equation}\label{eq:resolvlocGauss}
  \Ev{\abs{\ip{\mathbf{e}_i, (X_{W;N} - \lambda)^{-1}
  \mathbf{e}_j}}^s} \ \le \ A_s W^{\frac{s}{2}} \e^{-\alpha_s \frac{|i-j|}{W^{ 8}}}  , \quad \lambda \in [-r,r].
\end{equation}
 In particular, \eqref{eq:resolvlocGauss} holds for the Gaussian band ensemble \eqref{eq:GBE}.
\item If $\ln h$ and $\ln g$  are uniformly H\"older continuous with exponent $\alpha$,
then
given $r >0$  and $s  \in (0,1)$ there are $A_s < \infty$ and $\alpha_s > 0$  such
that
\begin{equation}\label{eq:resolvlocredux}
  \Ev{\abs{\ip{\mathbf{e}_i, (X_{W;N} - \lambda)^{-1}
  \mathbf{e}_j}}^s} \ \le \ A_s W^{\frac{s}{2}} \e^{-\alpha_s \frac{|i-j|}{W^{ \mu}}}  , \quad \lambda \in [-r,r],
\end{equation}
with $\mu= 5 + \frac{4}{\alpha}$
\item If $h$ and $g$ are proportional to characteristic functions of open neighborhoods of the origin, then  given  $r >0$  and $s  \in (0,1)$ there are $A_s < \infty$ and $\alpha_s > 0$  such
that \eqref{eq:resolvlocredux} holds with $\mu=9.$
\end{enumerate}
\end{thm}

 \appendix

 \section{A lemma on conditional averages}
 In the proofs of the various versions of Lemma F above, a key step was to estimate averages of the form \begin{equation}
 \label{eq:coupleexponential} \Ev{\e^{-\sum_{j=1}^{n}U_{j}}}
 \end{equation} 
 in which $U_{j}$ are non-negative, strictly positive with good probability, but not independent.  The following Lemma gives the relevant estimate, which can be seen as a simple version of \emph{stochastic domination}.  As the proof shows, under appropriate assumptions, we can estimate
 \eqref{eq:coupleexponential} in terms of the same expression with $U_{j}$ replaced by i.i.d.\ non-negative Bernoulli variables taking $0$ with probability less than $1$. 
 \begin{lem}\label{lem:prob} Let $\Sigma_j$ be a sequence of $\sigma$-algebras of events on a probability space and let $U_j$ be a sequence of non-negative random variables with $U_{j}$ measurable with respect to $\Sigma_{k}$ for $k \neq j$.  If for some $\delta > 0$,
 $$ \Pr(U_j \ge \delta | \Sigma_{j}) \ge p_0 $$
 for each $j$, then
 $$\Ev{\e^{-\sum_{j=1}^n U_j}} \ \le \ \e^{-(1-\e^{-\delta}) p_0 n}.$$
 \end{lem}
 \begin{proof}
 This follows by induction, since
 $$\Ev{\e^{-\sum_{j=1}^n U_j} | \Sigma_{n}} =\e^{-\sum_{j=1}^{n-1} U_{j}}
 \Ev{\e^{-U_{n}}| \Sigma_{n}} \le [(1-p_0) + \e^{-\delta} p_0] \e^{-\sum_{j=1}^{n-1} U_j}$$
 and 
 $$ (1-p_0) + \e^{-\delta} p_0 \le \e^{-(1-\e^{-\delta}) p_0}. \qedhere$$
 \end{proof}
 % \bibliography{references}
 \def\rmp{Rev.\ Math.\ Phys.\ }
\def\cmp{Comm.\ Math.\ Phys.\ }
\def\lmp{Lett.\ Math.\ Phys.\ }
\def\jsp{J.\ Stat.\ Phys.\ }
\def\prl{Phys.\ Rev.\ Lett.\ }
\def\pre{Phys.\ Rev.\ E }
\providecommand{\bysame}{\leavevmode\hbox to3em{\hrulefill}\thinspace}
\providecommand{\MR}{\relax\ifhmode\unskip\space\fi MR }
% \MRhref is called by the amsart/book/proc definition of \MR.
\providecommand{\MRhref}[2]{%
  \href{http://www.ams.org/mathscinet-getitem?mr=#1}{#2}
}
\providecommand{\href}[2]{#2}

\end{document}